\documentclass[12pt,english]{article}
\usepackage{graphicx}
\usepackage{array}
\usepackage[a4paper]{geometry}
\usepackage[unicode=true,bookmarks=true,bookmarksnumbered=true,bookmarksopen=true,
  bookmarksopenlevel=1,breaklinks=true,pdfborder={0 0 1},
  colorlinks=false,hyperfootnotes=false]{hyperref}
\hypersetup{
  bookmarksnumbered=true,
  colorlinks=true,
  linkcolor=blue,
  citecolor=blue,
  urlcolor=blue,
  menucolor=blue
}
\usepackage[natbibapa]{apacite}
\usepackage{longtable}
\usepackage{booktabs}
\usepackage{ragged2e}
\RaggedRight
\usepackage{parskip}
\usepackage{amsmath,amssymb,amsthm}
\usepackage{times}
\usepackage{setspace}
\usepackage[skip=0pt]{caption}
\usepackage{authblk}
\usepackage{enumitem}
\usepackage{tikz}
\usepackage{xcolor}
\usetikzlibrary{arrows.meta,positioning,shapes.geometric,calc}

\definecolor{plan}{RGB}{70,130,180}
\definecolor{audit}{RGB}{60,160,100}
\definecolor{decide}{RGB}{210,120,40}
\definecolor{improve}{RGB}{160,60,120}
\definecolor{textgray}{RGB}{80,80,80}

\newtheorem{proposition}{Proposition}[section]
\newtheorem{corollary}{Corollary}[section]
\newtheorem{remark}{Remark}[section]
\theoremstyle{definition}
\newtheorem{assumption}{Assumption}[section]
\newtheorem{example}{Example}[section]

\newcommand{\Cov}{\mathrm{Cov}}
\newcommand{\Var}{\mathrm{Var}}
\newcommand{\E}{\mathbb{E}}
\newcommand{\R}{\mathbb{R}}
\newcommand{\plim}{\overset{p}{\to}}
\newcommand{\dlim}{\overset{d}{\to}}

\DeclareMathOperator*{\argmin}{arg\,min}

\numberwithin{equation}{section}

\title{Fairness Testing for Algorithmic Pricing}

\author[1]{Fei Huang\footnote{Correspondence: Fei Huang, feihuang@unsw.edu.au. Giles Hooker, ghooker@wharton.upenn.edu}}
\author[2]{Giles Hooker}
\affil[1]{School of Risk and Actuarial Studies, UNSW Sydney}
\affil[2]{Department of Statistics and Data Science, Wharton School, University of Pennsylvania}
\date{}

\begin{document}
\onehalfspacing
\maketitle

\begin{abstract}
Algorithmic systems now set prices across auto insurance, credit, and lending markets, and regulators increasingly require firms to demonstrate that these systems do not discriminate against protected groups. The standard audit regresses pricing output on a protected attribute and legitimate rating factors, then tests the resulting coefficient using ordinary least squares standard errors. We show that this approach is structurally invalid. Pricing algorithms are usually deterministic, so residuals reflect approximation error rather than sampling variability, rendering classical standard errors invalid in both direction and magnitude. We derive correct asymptotic variance estimators for OLS and GLM audit regressions and the correct cross-covariance formula for proxy discrimination testing. Applied to quoted premiums from 34 Illinois auto insurers, every insurer fails the conditional demographic parity test, with minority zip codes paying $34–$158 more per year than comparable-risk white zip codes. The standard proxy discrimination formula flags zero insurers. However, our corrected formula identifies all 34 as statistically significant, of which 16 exceed the substantive threshold. Our framework provides statistically valid audit tools for any deterministic algorithmic system subject to regression-based fairness testing.
\end{abstract}


\section{Introduction}
\label{sec:intro}

Auto insurance, mortgage rates, credit limits, and
hiring decisions are increasingly set by algorithms rather
than people. When those algorithms disadvantage minority groups relative to
comparable majority groups, the consequences are concrete.
Affected households pay higher premiums for the same risk,
face higher interest rates for the same creditworthiness,
and receive fewer callbacks for equivalent job applications, prompting regulators to respond with new disclosure and testing requirements. Colorado and New
York State now require firms to demonstrate that
their pricing systems do not discriminate against protected
groups \citep{doradoi-2023-concerning,
nydfs-2024-proposed}, and a growing literature in
information systems documents the reality of such
discrimination in practice \citep{fu-etal-2021-crowds,
hu-etal-2025-human, zhang-xu-2024-fairness,
hurlin-etal-2025-fairness}. What this literature has not
examined is whether the statistical methods used to detect
it are valid.

We show that they are not valid. The standard audit regresses
pricing output on a protected attribute and a set of
legitimate rating factors, then tests whether the coefficient
is zero or within a tolerance band using ordinary least squares (OLS) or generalized linear model (GLM) standard
errors. Those standard errors assume the response is
stochastic. However, pricing algorithms are usually deterministic. The same
profile always returns the same price, so residuals reflect
approximation error rather than sampling variability. The
classical variance formula $\hat\sigma^2(X^\top X)^{-1}$ is
therefore structurally invalid, and its bias is unpredictable
in direction without computing the correct estimator.

These inferential errors have large practical consequences.
As we show in Section~\ref{sec:empirical}, applying the
corrections to quoted premiums from 34 Illinois auto
insurers shifts standard error ratios by as much as a
factor of 2.6 in either direction within the same audit.
The proxy discrimination test is more severely affected
still. Testing for proxy discrimination requires comparing
a coefficient across two regressions on the same
deterministic output vector. Standard practice treats
these regressions as independent, which overstates the
variance of the shift and makes the test conservative by
construction, suppressing the detection of proxy
discrimination that the corrected formula would reveal.

Three further gaps compound these inferential failures,
and they appear in the scholarly literature as prominently
as in regulatory practice. First, information systems (IS) and management science (MS) studies of algorithmic fairness test for disparity using criteria drawn from the machine learning (ML) classification literature (such as independence, separation, or sufficiency) without connecting them to the named regulatory criteria that govern legal compliance in pricing and lending \citep{hurlin-etal-2025-fairness, zhang-xu-2024-fairness, fu-etal-2021-crowds}. Because fairness criteria are mutually incompatible in general \citep{kleinberg-etal-2018-inherent, barocas-etal-2022-fairness}, applying a criterion from the ML literature without first verifying that it aligns with the specific regulatory standard being enforced can produce findings that are irrelevant to legal compliance. 

Second, the fairness testing literature relies on
significance testing to evaluate compliance
\citep{hurlin-etal-2025-fairness,
fuster-etal-2022-predictably}. The likelihood-ratio
tests of \citet{hurlin-etal-2025-fairness} exemplify
this approach. However, significance testing cannot
provide positive evidence of fairness. Failing to reject the null of fairness establishes only that the data do not provide sufficient evidence of a violation, not that the model affirmatively complies with regulatory tolerance. With sufficiently large samples, any non-zero gap will be flagged as significant even when it falls well within acceptable limits. With small samples, genuine violations may go entirely undetected.

Third, fairness testing in both the literature and regulatory proposals treats statistical testing, sample design, and remediation as separate exercises rather than an integrated pre-committed pipeline \citep{doradoi-2023-concerning, nydfs-2024-proposed, hurlin-etal-2025-fairness}. We address this by providing a complete audit protocol in which all design choices are pre-specified, the decision rule is derived from the inferential results in Sections 3 and 4, and the three-outcome rule (Pass, Fail, Insufficient Information) explicitly accounts for the difference between a model that fails and one for which the sample is too small to reach a verdict.
 
This paper addresses the inferential problem and all
three methodological gaps. We derive the correct
asymptotic variance for OLS and GLM estimators applied
to deterministic algorithm outputs
(Propositions~\ref{prop:ols} and~\ref{prop:glm}) and the
correct cross-covariance formula for the coefficient-shift
test used to detect proxy discrimination
(Proposition~\ref{prop:shift}). Beyond the inferential
fix, we connect two fairness criteria --- proxy
discrimination and conditional demographic parity --- to
named statistical estimands and regulatory sources, adopt
equivalence testing (TOST) \citep{schuirmann-1987-comparison}
as the decision framework, and combine these elements into
a complete audit protocol with pre-specified design,
criterion-specific testing, and a three-outcome decision
rule that distinguishes Pass, Fail, and Insufficient
Information.
 
We validate the framework using publicly available quoted
premiums from 34 Illinois auto insurers. As discussed in
Section~\ref{sec:empirical}, the data operate at the
zip-code level and use proxied protected attributes, so
the results serve as an illustration of the framework
rather than a definitive audit of any individual insurer.
Our results show that every insurer fails the conditional demographic parity test.
Minority zip codes pay \$34--\$158 more per year than
comparable-risk white zip codes after controlling for risk
and geography.  The inferential correction also transforms
the proxy discrimination results. The standard formula
flags zero of 34 companies, while the corrected formula
finds 34 statistically significant cases, of which 16
also exceed the substantive shift threshold and are
formally flagged. The
framework applies to any deterministic algorithmic system
subject to regression-based fairness testing, including
credit scoring, mortgage pricing, and hiring screening.
 
Section~\ref{sec:background} reviews the related literature
and positions our contribution. Section~\ref{sec:inference}
derives the correct inferential foundations.
Section~\ref{sec:framework} develops the fairness criteria,
TOST decision framework, and audit protocol.
Section~\ref{sec:empirical} reports the empirical validation.
Section~\ref{sec:discussion} discusses implications and
Section~\ref{sec:conclusion} concludes.
\section{Related Literature}
\label{sec:background}

\subsection{Algorithmic bias and fairness in IS and management science}

\citet{fu-etal-2021-crowds} find evidence of racial and
gender bias in peer-to-peer lending algorithms even when
protected attributes are excluded as inputs, and propose
a debiasing method to address it.
\citet{hu-etal-2025-human} decompose human evaluator bias
in microloan granting into preference-based and
belief-based components and examine how these biases
evolve when ML algorithms are trained on the resulting
decisions. They find that even fairness-unaware algorithms
partially reduce the human bias present in the training
data, though removing the bias from training data improves
fairness further. \citet{fuster-etal-2022-predictably}
find that fintech mortgage lenders produce racial and
ethnic pricing disparities as large as those of
traditional lenders, a result that undercuts the claim
that algorithmic systems reduce discrimination by removing
human discretion. All three studies establish that algorithmic
discrimination is a pressing problem in
the lending and mortgage markets that fairness audits
 and testing are needed to police.

 A separate strand of the management science literature warns that
imposing fairness constraints can backfire. When firms
anticipate impact-parity requirements, they reduce
investment in model accuracy in exactly the settings
where better models would matter most
\citep{fu-aseri-2022-unfair, shimao-etal-2025-strategic}.
This makes valid auditing more consequential. If fairness tests are unreliable, the constraints
they trigger may harm the groups they protect.

\citet{zhang-xu-2024-fairness} document disparate impact
in catastrophic insurance ratemaking and propose a
fair ratemaking solution. \citet{hurlin-etal-2025-fairness}
develop a diagnostic framework for credit scoring that
identifies which variables drive a lack of fairness.
Both papers address the design of the algorithm rather
than the validity of the statistical inference used to
evaluate it. Neither examines whether the tests applied
to algorithm outputs are valid when those outputs are
deterministic. 

\citet{lambrecht-tucker-2019-algorithmic} find that a
gender-neutral ad delivery algorithm systematically
showed STEM career ads to fewer women, because
optimising for cost-effective ad delivery drew on the
higher market price of female views driven by competing
advertisers, a structural feature of the ad auction
rather than any discriminatory intent. \citet{rhue-2023-anchoring}
finds that emotion recognition algorithms exhibit
systematic demographic disparities in scoring, and that
human labellers anchor on these biased scores even when
explicitly informed of the algorithm's fairness
limitations, providing no evidence that disclosure alone
offsets the bias. Together, these results illustrate that
algorithmic bias can arise from market structure and
cognitive mechanisms rather than from flawed design,
and that transparency interventions do not reliably
correct it. What remains essential, therefore, is the ability to detect such disparities through statistically valid audit tools. We show that the regression-based methods currently used to audit deterministic algorithms are structurally invalid, meaning that the disparities these studies document may go undetected or be incorrectly measured even when regulators attempt to test for them.

\subsection{Inference for deterministic responses}

\citet{eicker-1967-limit} and \citet{white-1980-heteroskedasticity} develop the heteroskedasticity-consistent sandwich estimator to correct for non-constant error variance in stochastic regression models. In the deterministic setting we analyze, the sandwich is required for a different reason. The residuals of a linear audit model fitted to algorithm outputs reflect approximation error rather than draws from a distribution, invalidating the classical variance formula regardless of whether the error variance is constant.

\citet{xin2026} identify a related but distinct distortion in regression-based fairness audits. Using a proxied rather than directly observed race variable introduces measurement error into the protected attribute, biasing the estimated disparity coefficient itself. Our paper identifies a compounding inferential problem that operates independently of measurement error. Even when the protected attribute is correctly measured, the standard errors of the regression coefficients are invalid because the response is deterministic rather than stochastic. Both distortions can operate simultaneously in practice, and the corrections we derive remain essential even when race is accurately observed.


This inference problem is related to, but distinct from,
two other non-standard regression settings. In the
generated regressors literature \citep{pagan-1984-econometric,
murphy-topel-1985-estimation}, a first-stage estimate
appears as a regressor in a second-stage equation, and
the variance correction accounts for the estimation error
in the first stage. In our setting the deterministic
function appears as the response, not a regressor, so
the source of non-standard variance is different and
the correction takes a different form. The connection
to White's misspecification-robust inference
\citep{white-1982-maximum} is closer in spirit but still
not exact. White's framework applies to stochastic models
whose distributional assumptions fail, whereas our setting
has no distributional assumption to fail. There is simply
no error term.

The cross-covariance correction for the proxy
discrimination test has particularly large empirical
consequences. The coefficient-shift test compares a coefficient across two regressions, one omitting the protected attribute and one including it, following the logic of the Hausman specification test  \citep{hausman-1978-specification}. Under that framework, estimates from two consistently estimated models are asymptotically uncorrelated, a condition that holds when the response is stochastic. When both models are fitted to the same deterministic response vector, however, this condition fails by construction, and the standard independent-samples variance formula consequently overstates the variance of the shift.

\citet{schuirmann-1987-comparison} introduced equivalence
testing (TOST) in pharmaceutical regulation to demonstrate
bioequivalence rather than merely the absence of a
detectable difference. The logic transfers directly to
fairness auditing. The tolerance thresholds already
embedded in regulatory proposals define an equivalence
region, and TOST places the burden of demonstrating
compliance on the firm rather than the regulator.

\subsection{Fairness criteria and regulatory proposals}

The two fairness criteria we study operationalise
the legal doctrine of disparate impact
\citep{barocas-moritz-2016-big}, which holds that facially
neutral policies producing disproportionate harm to
protected groups are discriminatory even absent the
explicit use of protected attributes. The machine learning
fairness literature has formalised several statistical
criteria under this doctrine and shown that no system can
simultaneously satisfy several of them unless base rates
are equal across groups \citep{kleinberg-etal-2018-inherent,
barocas-etal-2022-fairness, mehrabi-etal-2021-survey}. This incompatibility means that the choice of criterion
embeds a normative commitment, and that practitioners who
test for one criterion while implicitly targeting another
may reach verdicts that conflict when the two criteria
diverge on the same data.
 
We focus on two criteria in this paper, proxy discrimination
(PD) and conditional demographic parity (CDP). These are
motivated by both regulatory proposals and conceptual
coherence. Regulatorily, the test statistics already
specified by Colorado \citep{doradoi-2023-concerning} and
New York \citep{nydfs-2024-proposed} directly
operationalise PD and CDP, even though none of these
proposals names the criterion being tested. Conceptually,
PD corresponds to the indirect discrimination doctrine, that is
testing whether a rating variable acts as a statistical
proxy for the protected attribute. CDP corresponds
to the requirement that premiums be equal across groups
after conditioning on legitimate risk factors.
 
Within insurance and actuarial science, a parallel
literature has formalised fairness criteria for pricing
models. \citet{lindholm2022} develop the concept of
discrimination-free pricing, and \citet{Frees2023}
survey the landscape of fairness criteria as they apply
to actuarial practice. \citet{xin2022} connect machine
learning fairness criteria to antidiscrimination
regulations in insurance pricing, and related approaches
include fairness-constrained ratemaking
\citep{vincent2022fair, araiza2022discrimination,
cote2024fair, HENCKAERTS2022117230}. A critical
limitation of this literature, documented by
\citet{shimao-huang-2024-welfare} and
\citet{shimao-huang-2025-welfare}, is that fairness
constraints applied at the cost modelling stage do not
necessarily produce fair prices in the market. Demand
responses and price optimisation can undo the fairness
properties of the underlying cost model. Our paper
operates at the pricing output stage, by auditing the
premiums that consumers actually face. We show that the standard errors used in these regression-based tests are structurally invalid when the response is a deterministic algorithm output.
 
\citet{pope-sydnor-2011-implementing} develop a framework
for identifying and eliminating proxy discrimination in
statistical profiling models, illustrated using employment
profiling data. Their framework uses the coefficient-shift
approach to detect whether a variable acts as a proxy for
a protected attribute, then proposes a method to remove
the proxy effect while preserving predictive accuracy.
We adopt the detection component of this framework and
derive the correct variance for the coefficient shift
when both regressions are fitted to a deterministic
algorithm output, a correction that is not required in
the stochastic setting Pope and Sydnor analyse.

\section{Statistical Inference for Deterministic Algorithm Audits}
\label{sec:inference}
 
\subsection{Setup and notation}
 
Let $f: \R^q \to \R$ denote a pricing algorithm that maps a
vector of rating variables $z \in \R^q$ to a price $f(z) \in \R$.
We assume $f$ is a fixed, known function. Given $z$, $f(z)$ is
determined without any additional randomness. This models the
typical audit setting in which the auditor submits a profile and
records the returned quote.
 
Let $\mathcal{D}_n = \{(x_i, z_i)\}_{i=1}^n$ denote a sample
of $n$ observations, where $x_i \in \R^p$ is a vector of audit
covariates (rating factors and protected attribute indicators) and
$z_i \in \R^q$ is the pricing profile submitted for observation
$i$. In some applications $z_i = x_i$, but the pricing function
may also use variables that the auditor does not observe. We
allow $x_i$ and $z_i$ to overlap or coincide. The auditor
fits the linear model
\begin{equation}
  f(z_i) = x_i^\top \beta + r_i, \qquad i = 1, \ldots, n,
  \label{eq:linear}
\end{equation}
where $r_i = f(z_i) - x_i^\top \beta$ is the approximation
error, the difference between the true pricing function and its
linear projection onto the audit covariates. Crucially, $r_i$
is a deterministic function of $(x_i, z_i)$, not a random
variable independent of $x_i$. The population target of the
OLS estimator is
\begin{equation}
  \beta^* = \argmin_{\beta \in \R^p}
    \E\!\left[(f(z) - x^\top \beta)^2\right],
  \label{eq:target}
\end{equation}
the linear projection of $f(z)$ onto the span of the audit
covariates. We write $X = (x_1, \ldots, x_n)^\top \in \R^{n \times p}$
for the covariate matrix and $F = (f(z_1), \ldots, f(z_n))^\top
\in \R^n$ for the response vector. The OLS estimator is
$\hat\beta = (X^\top X)^{-1} X^\top F$.
 
\begin{assumption}[Sampling]
\label{ass:sampling}
The observations $\{(x_i, z_i)\}_{i=1}^n$ are i.i.d.\ draws
from a distribution $\mathbb{P}$ on $\R^p \times \R^q$ with
$\E[\|x\|^4] < \infty$ and $\E[|f(z)|^4] < \infty$.
\end{assumption}
 
\begin{assumption}[Identification]
\label{ass:ident}
The matrix $\Sigma_{xx} = \E[x x^\top]$ is positive definite.
\end{assumption}
 
\begin{assumption}[Bounded approximation error]
\label{ass:approx}
$\E\!\left[(f(z) - x^\top\beta^*)^2 \|x\|^2\right] < \infty$.
\end{assumption}
 
Assumption~\ref{ass:sampling} treats the audit profiles as
random draws, appropriate when they are sampled from the
portfolio or a designed test mix. Assumption~\ref{ass:ident}
a standard identifiability condition. Assumption~\ref{ass:approx} is a mild moment
condition that holds whenever $f$ and $x$ have bounded fourth
moments.
 
\subsection{The measurement validity problem: incorrect standard errors}
 
The classical OLS formula uses $\hat\sigma^2(X^\top X)^{-1}$,
where $\hat\sigma^2 = n^{-1}\sum \hat r_i^2$. This formula
is derived under the assumption that $r_i$ is an i.i.d.\ error
term independent of $x_i$ with constant variance. For
deterministic pricing algorithms, neither condition holds.
The approximation errors $r_i$ reflect the nonlinearity of
$f$ relative to the linear audit model. They are deterministic
functions of the covariates, not independent draws from a
distribution.
 
\begin{proposition}[OLS consistency and asymptotic variance]
\label{prop:ols}
Under Assumptions~\ref{ass:sampling}--\ref{ass:approx}, the OLS estimator
$\hat\beta = (X^\top X)^{-1} X^\top F$ satisfies
$\hat\beta \plim \beta^*$ and
\[
  \sqrt{n}(\hat\beta - \beta^*)
    \dlim \mathcal{N}\!\left(0,\, \Sigma_{xx}^{-1}
    \Omega \Sigma_{xx}^{-1}\right),
\]
where $\Omega = \E[x x^\top (f(z) - x^\top \beta^*)^2]$.
A consistent estimator of the asymptotic covariance is
\[
  \widehat{\Cov}(\hat\beta)
    = (X^\top X)^{-1}
      \left(\sum_{i=1}^n x_i x_i^\top \hat r_i^2\right)
      (X^\top X)^{-1},
\]
where $\hat r_i = f(z_i) - x_i^\top\hat\beta$ are the OLS
residuals.
\end{proposition}
\noindent\textit{Proof.} See Appendix~\ref{app:proofs}, Section~\ref{proof:ols}. \hfill$\square$
 
This is the HC0 heteroskedasticity-consistent estimator of
\citet{white-1980-heteroskedasticity}, but the justification
in the deterministic setting is different from the usual
heteroskedasticity rationale. The residuals $\hat{r}_i$ reflect
the nonlinearity of $f$ relative to the linear audit model,
not the variance of a stochastic error term. The finite-sample 
corrections HC1 through HC3 of \citet{mackinnon-white-1985} 
are consistent estimators of the same asymptotic sandwich 
variance in the deterministic setting and may be used in place 
of HC0. Specifically, HC3 inflates each squared residual by 
the leverage correction $(1-h_{ii})^{-2}$, giving
\begin{equation}
    \widehat{\operatorname{Cov}}_{\mathrm{HC3}}(\hat{\beta}) 
    = (X^{\top}X)^{-1}
    \left(
        \sum_{i=1}^{n} x_i x_i^{\top} 
        \frac{\hat{r}_i^2}{(1-h_{ii})^2}
    \right)
    (X^{\top}X)^{-1},
    \label{eq:hc3}
\end{equation}
where $h_{ii} = x_i^{\top}(X^{\top}X)^{-1}x_i$ is the 
leverage of observation $i$. Under Assumption~3.2, 
$h_{ii} \in (0,1)$ for all $i$ in sufficiently large 
samples, ensuring the HC3 correction is well-defined. 
HC3 provides better finite-sample performance than HC0 
by down-weighting high-leverage observations, a property 
that is desirable when audit profiles are not uniformly 
distributed across the covariate space. We use HC3 
throughout, following standard practice.
 
The classical formula is valid if and only if the approximation
errors are uncorrelated with the covariates and have constant
variance. For the algorithmic pricing systems targeted by
modern fairness regulation, this condition is unlikely to hold.
 
\begin{example}[Territorial loading]
\label{ex:territorial}
Suppose $f(z) = \exp(\lambda^\top z)$ where $z$ includes
individual rating factors and territory indicators, and
$\lambda$ partitions conformably into territory and
individual rating coefficients. If minority zip codes are
concentrated in certain territories, then the approximation
error $r_i = f(z_i) - x_i^\top\beta^*$ is correlated with
the minority indicator $A_i$, so $\Cov(A_i, r_i^2) \neq 0$
and the classical formula is incorrect. The sandwich
estimator correctly accounts for this.
\end{example}
 
\begin{corollary}
\label{cor:agree}
The classical OLS covariance formula is valid if and only if
$\Cov(x, r^2) = 0$, that is, the second moment of the
approximation error is uncorrelated with the covariates. A
sufficient condition is that $f(z) = x^\top\beta_0 + \varepsilon$
for some $\beta_0$ and some error $\varepsilon$ independent of
$x$, which is the case when the linear audit model is correctly
specified.
\end{corollary}
 
For the algorithmic pricing systems targeted by modern fairness
regulation, correct specification is unlikely. Territory
relativities, credit tiers, and vehicle classification
structures all create approximation errors that are correlated
with the protected attributes being tested.
 
\subsection{Extension to GLMs}
 
Pricing audits sometimes use a GLM, for example a 
log-linear model for right-skewed insurance premiums. 
Standard GLM theory gives
\begin{equation}
    \widehat{\operatorname{Cov}}(\hat{\beta}_{\mathrm{GLM}}) 
    = (X^{\top}\hat{\Lambda}X)^{-1},
    \label{eq:glm-cov}
\end{equation}
where $\hat{\Lambda} = \operatorname{diag}(\hat{\lambda}_1, 
\ldots, \hat{\lambda}_n)$ is the diagonal matrix of 
estimated GLM working weights, with
\begin{equation}
    \hat{\lambda}_i = \frac{1}{V(\hat{\mu}_i)\cdot 
    \{g'(\hat{\mu}_i)\}^2},
    \label{eq:glm-weights}
\end{equation}
$V(\hat{\mu}_i)$ the variance function of the posited 
exponential family evaluated at the fitted mean 
$\hat{\mu}_i = g^{-1}(x_i^\top\hat{\beta}_{\mathrm{GLM}})$, 
and $g'(\hat{\mu}_i)$ the derivative of the link 
function $g(\cdot)$ evaluated at $\hat{\mu}_i$. This 
formula is derived under the assumption that the 
response $y_i$ is a genuine draw from the posited 
exponential family, so that the expected Hessian and 
the outer product of scores are equal. When $f(z_i)$ 
is a deterministic algorithm output rather than a 
genuine draw, this equality fails and 
$(X^{\top}\hat{\Lambda}X)^{-1}$ is no longer a valid 
variance estimator, as we show in 
Proposition~\ref{prop:glm}.
 
\begin{proposition}[GLM consistency and asymptotic variance]
\label{prop:glm}
Under regularity conditions analogous to
Assumptions~\ref{ass:sampling}--\ref{ass:approx} and standard
GLM smoothness conditions, the GLM estimator is consistent for
\[
  \beta^* = \argmin_\beta \E\!\left[-\log p(f(z) \mid x^\top\beta)\right],
\]
and satisfies $\sqrt{n}(\hat\beta^{\mathrm{GLM}} - \beta^*)
\dlim \mathcal{N}(0, J^{-1} M J^{-1})$, where
\[
  J = \E\!\left[x x^\top
    \left.\frac{-\partial^2}{\partial\eta^2}
    \log p(f(z) \mid \eta)\right|_{\eta = x^\top\beta^*}\right]
\]
and
\[
  M = \Cov\!\left(x\left.
    \frac{\partial}{\partial\eta}
    \log p(f(z) \mid \eta)\right|_{\eta = x^\top\beta^*}\right).
\]
In the standard GLM where $y_i$ is a  draw from
$p(\cdot \mid \eta_i)$, the information identity gives $J = M$
and the sandwich reduces to $(X^\top \hat W X)^{-1}$. When
$f(z_i)$ is deterministic, $J \neq M$ in general and the
standard formula is incorrect.
\end{proposition}
\noindent\textit{Proof.} See Appendix~\ref{app:proofs}, Section~\ref{proof:glm}. \hfill$\square$
 
\subsection{The cross-covariance problem in proxy discrimination}
 
The proxy discrimination test compares a coefficient $\hat\phi$
on a suspected proxy variable across two regression models, a
restricted model omitting the protected attribute $A$ and an
extended model including it. Denoting the covariate matrices as
$X \in \R^{n \times p}$ (restricted) and $\tilde X \in \R^{n
\times (p+1)}$ (extended), the two estimates of interest are
\begin{equation}
  \hat\phi = e_j^\top (X^\top X)^{-1} X^\top F,
  \qquad
  \hat\phi' = e_k^\top (\tilde X^\top \tilde X)^{-1}
              \tilde X^\top F,
  \label{eq:two-estimates}
\end{equation}
where $e_j$ and $e_k$ select the coefficient on the proxy
variable in each model. Writing $a^\top = e_j^\top
(X^\top X)^{-1} X^\top$ and $\tilde a^\top = e_k^\top
(\tilde X^\top \tilde X)^{-1} \tilde X^\top$, we have
\begin{equation}
  \hat\phi - \hat\phi' = (a - \tilde a)^\top F.
  \label{eq:diff-linear}
\end{equation}
Standard practice treats the two estimates as independent and
sums their variances, following the logic of the Hausman specification test \citep{hausman-1978-specification}. When both models are fitted to the same
deterministic response vector $F = f(Z)$, however, $(a^\top F,
\tilde a^\top F)$ are correlated through the shared response
and the independence assumption fails.
 
\begin{proposition}[Variance of the coefficient shift]
\label{prop:shift}
Let $X$ be the restricted covariate matrix and $\tilde X$
the extended matrix. Under Assumptions~\ref{ass:sampling}--\ref{ass:approx}, the asymptotic
variance of $\sqrt{n}(\hat\phi - \hat\phi')$ is
\begin{align}
  \Var(\hat\phi - \hat\phi')
    &= (X^\top X)^{-1} \Cov(X_j\, f(Z))(X^\top X)^{-1}
     \notag \\
    &\quad + (\tilde X^\top \tilde X)^{-1}
       \Cov(\tilde X_k\, f(Z))(\tilde X^\top \tilde X)^{-1}
     \notag \\
    &\quad - 2\,(X^\top X)^{-1}
       \Cov(X_j\, f(Z),\, \tilde X_k\, f(Z))
       (\tilde X^\top \tilde X)^{-1}.
  \label{eq:var-shift}
\end{align}
The cross-covariance term is non-zero because both models
are evaluated at the same $f(z_i)$, making their coefficient
estimates positively correlated. The independent-samples
formula omits this term and overstates the variance,
rendering the proxy discrimination test systematically
conservative. All three covariance matrices can be estimated
directly from the data:
\begin{align*}
  \widehat{\Cov}(X_j\, f(Z))
    &= \frac{1}{n}\sum_{i=1}^n x_{ij}^2 f(z_i)^2
       - \left(\frac{1}{n}\sum_{i=1}^n x_{ij} f(z_i)\right)^2, \\
  \widehat{\Cov}(X_j f(Z),\, \tilde X_k f(Z))
    &= \frac{1}{n}\sum_{i=1}^n x_{ij}\tilde x_{ik} f(z_i)^2
       - \left(\frac{1}{n}\sum x_{ij} f(z_i)\right)
         \left(\frac{1}{n}\sum \tilde x_{ik} f(z_i)\right).
\end{align*}
\end{proposition}
\noindent\textit{Proof.} See Appendix~\ref{app:proofs}, Section~\ref{proof:shift}. \hfill$\square$
 
\begin{remark}
The sign of the cross-covariance term depends on
$\Cov(X_j f(Z),\, \tilde X_k f(Z))$. When this is positive,
which is the typical case when the extended model adds only
the protected attribute $A$ to the restricted model, the
independent-samples formula over-estimates the variance of
the shift and the test is conservative. When it is negative,
the test is anti-conservative. Computing the full formula is
therefore essential for reliable inference in either direction.
\end{remark}
 
In practice, two audits of the
same algorithm are not independent measurements. When both
models share a deterministic response, treating them as
independent inflates the standard error of the shift, making
the test less likely to detect proxy discrimination. We show
in Section~\ref{sec:empirical} that this effect reclassifies
17 of 34 Illinois insurers from non-significant to significant.
 
\subsection{When does the correction matter?}
\label{subsec:correction-magnitude}
 
The practical importance of the sandwich correction depends
on how well the linear audit model approximates the true
pricing function. A useful diagnostic is the ratio
\begin{equation}
  \rho_j = \frac{
    [\widehat{\Cov}_{\mathrm{sandwich}}(\hat\beta)]_{jj}^{1/2}
  }{
    [\widehat{\Cov}_{\mathrm{classical}}(\hat\beta)]_{jj}^{1/2}
  },
  \label{eq:rho}
\end{equation}
the ratio of sandwich to classical standard error for the
$j$-th coefficient. When $\rho_j \approx 1$, the pricing
function is well approximated by the linear model and the
classical formula is adequate. Departures from 1 in either
direction signal that the linear model is misspecified relative
to $f$ and that the classical variance should not be used.
 
Three scenarios determine the direction and magnitude of the
correction. When the pricing function is well-approximated
linearly, $r_i$ is small and mean-independent of $x_i$, so
$\rho_j \approx 1$ and the correction is negligible. When the
function is nonlinear but the approximation error is
mean-independent of the covariates, $\hat\beta$ remains
unbiased for $\beta^*$ and the sandwich estimator is
consistent. The classical formula over- or under-estimates
depending on the sign of $\Cov(x_j, r^2)$, which cannot be
determined without computing both estimators. When the
approximation error is correlated with the covariates,
$\hat\beta$ remains consistent for $\beta^*$ (the best linear
approximation) but the classical standard error is incorrect
in an unpredictable direction. This third scenario is the
empirically relevant case for insurance pricing algorithms.
As Section~\ref{sec:empirical} documents, territory
relativities, credit tiers, and other non-linear pricing
components typically correlate with protected attributes,
making the correction both necessary and consequential.
 
\section{A Fairness Audit Framework}
\label{sec:framework}
 
\subsection{Fairness criteria as testable hypotheses} \label{sec:FairnessHypo}
 
The inferential results in Section~\ref{sec:inference}
establish which standard error estimators are valid for
regression-based fairness tests on deterministic algorithm
outputs. This section builds on those foundations to
formalise the fairness criteria, specify the decision
framework, and assemble the complete audit protocol.
Existing regulatory proposals specify test statistics
without naming the fairness criterion they operationalise.
We formalise two criteria, proxy discrimination and conditional demographic parity, that are directly tested by the
regulatory proposals and tractable with quoted-premium
data. Each connects a test statistic to a named fairness
concept and a statistical estimand, using the corrected
inference of Propositions~\ref{prop:ols}--\ref{prop:shift}.
 
Throughout this section, $P$ denotes the pricing outcome
for a single observation, that is $P_i = f(z_i)$, and
$F = (P_1, \ldots, P_n)^\top$ denotes the stacked response
vector used in the matrix algebra of Section~3. $A$ denotes the binary
protected attribute, and $X_\ell$ denotes the vector
of approved legitimate rating factors.
 
\paragraph{Proxy Discrimination (PD).}
Following \citet{doradoi-2023-concerning}, 
\citet{pope-sydnor-2011-implementing}, and \citet{prince-schwarcz-2020-proxy}, the PD criterion
asks whether a variable $W_j$ acts as a statistical
substitute for the protected attribute. Please refer to Appendix~\ref{app:pd} for details of the criterion and its alignment with the regulation.

The two regression
models underlying the test are the restricted model,
\begin{equation}
  P = \mu + \phi_j W_j + \gamma^\top X_\ell + \varepsilon,
  \label{eq:pd-restricted}
\end{equation}
which omits the protected attribute $A$, and the extended
model,
\begin{equation}
  P = \mu' + \phi'_j W_j + \kappa A + \gamma'^\top X_\ell
      + \varepsilon',
  \label{eq:pd-extended}
\end{equation}
which includes it, and where $\kappa$ is the coefficient 
on the protected attribute $A$. The estimand is the coefficient shift
\begin{equation}
  \Delta_{\mathrm{PD},j} = \phi_j - \phi'_j,
  \label{eq:pd-estimand}
\end{equation}
the change in the coefficient on $W_j$ when $A$ enters
the model. A non-zero shift reveals that $W_j$ was
absorbing part of the influence of $A$ in the restricted
model, the defining feature of proxy discrimination.
 
The
decision rule combines a statistical criterion with a
substantive magnitude threshold. The test statistic for
the coefficient shift is
\[
  z_{\mathrm{PD},j} = \frac{\hat\phi_j - \hat\phi'_j}{
    \widehat{\mathrm{SE}}(\hat\phi_j - \hat\phi'_j)},
\]
where $\widehat{\mathrm{SE}}(\hat\phi_j - \hat\phi'_j)$
is the square root of the full cross-covariance estimator
in Proposition~\ref{prop:shift}. A variable $W_j$ is
flagged as a proxy discriminator if (1) $|z_{\mathrm{PD},j}|
> z_{1-\alpha}$, and (2) the relative shift
$|(\hat\phi_j - \hat\phi'_j) / \hat\phi_j|$ exceeds
a pre-specified minimum $\rho_{\min}$, which we set at
10\%. This threshold is a researcher judgment, chosen to
require that the shift represent a substantively meaningful
change in the coefficient rather than a statistically
detectable but economically trivial movement. It should
be pre-specified and documented in any regulatory
application.
The first criterion guards against spurious detection in
large samples. The second guards against flagging shifts
that are statistically detectable but too small to
indicate meaningful proxy behaviour.
 
\paragraph{Conditional Demographic Parity (CDP).}
The CO DOI draft regulation \citep[][\S\S\,6--7]{doradoi-2023-concerning} encodes a two-step test to operationalise CDP. Please refer to Appendix~\ref{app:cdp} for details of the criterion and its alignment with the draft regulation.

After controlling for legitimate risk factors, premiums
should not differ systematically across protected groups.
The estimands are the conditional mean difference and ratio,
\[
\Delta_\mu = \E[P \mid X_\ell, A=a] - \E[P \mid X_\ell, A=b]
\quad \text{and} \quad
R_\mu = \frac{\E[P \mid X_\ell, A=a]}{\E[P \mid X_\ell, A=b]}.
\]
In its strict form, CDP requires equality of the entire
conditional distribution. In practice we operationalise a
mean-based relaxation, testing whether the conditional
expectation is equal across groups after controlling for
$X_\ell$. The general regression model is
\begin{equation}
  P_i = \mu_0 + \beta_A\,\mathbf{1}\{A_i = a\}
    + \gamma^\top X_{\ell,i} + \varepsilon_i,
  \label{eq:cdp}
\end{equation}
where $\mu_0$ is the intercept, and $\hat\beta_A$ estimates $\Delta_\mu$ directly. A
log-linear specification $\log P_i = \mu_0 +
\beta_A\,\mathbf{1}\{A_i = a\} + \gamma^\top X_{\ell,i}
+ \varepsilon_i$ is also natural, in which case $\beta_A
= \log R_\mu$, $\exp(\hat\beta_A)$ estimates $R_\mu$, and
the implied dollar gap at mean premium $\bar{P}$ is
$\bar{P}(\exp(\hat\beta_A) - 1)$. Because the pricing
output $P_i$ is a deterministic function of the submitted
profile, the standard error of $\hat\beta_A$ must be
computed using the HC3 sandwich estimator of
Proposition~\ref{prop:ols} regardless of which
specification is used. Note that for the same reason, adherence to classical regression assumptions does not need to dictate which form of model is appropriate. 
 
\subsection{Equivalence testing as the decision framework}
 
Conventional significance testing places the null at zero
disparity, so a model passes unless the data prove
otherwise. In large administrative datasets, any non-zero
gap will be detected, flagging models as unfair even when
the disparity falls within regulatory tolerance. In small
datasets, genuine disparities go undetected. Neither
outcome is useful for a compliance audit.
 
We adopt equivalence testing (TOST)
\citep{schuirmann-1987-comparison}, under which the null
hypothesis is that the disparity equals or exceeds the
regulatory tolerance. A model passes only when the
data provide affirmative evidence that the disparity is
within tolerance. This formulation places the burden of
demonstrating compliance on the firm, which is the
appropriate standard for a regulatory audit.
 
For the level-gap criterion, a model passes CDP with respect
to margin $\delta > 0$ if the $(1-2\alpha)$ confidence interval
for $\Delta_\mu$ lies entirely within $(-\delta, +\delta)$.
Equivalently, we test
\[
H_0: |\Delta_\mu| \ge \delta \quad \text{vs.} \quad
H_A: |\Delta_\mu| < \delta.
\]
For the ratio criterion with tolerance $\tau \in (0,1)$, the
test is
\[
H_0: |\log R_\mu| \ge \log(1/\tau) \quad \text{vs.} \quad
H_A: |\log R_\mu| < \log(1/\tau).
\]
In this paper, we assume a company passes CDP if and only if both conditions hold
simultaneously, that is the CI for $\hat\beta_A$ lies entirely within
$(\log\tau, -\log\tau)$ and the implied dollar gap lies
within $(-\delta, +\delta)$. In practice, regulators could set specific tolerance standards that meet their requirements. 
 
The regulatory tolerance margins already specified in existing
guidance map directly onto the TOST bounds. The CO DOI 5\%
price gap threshold corresponds to $\delta = 0.05 \times
\bar{P}$; the standard 0.80 adverse impact ratio corresponds
to $\tau = 0.80$, giving a log-ratio band of $(-0.223, +0.223)$.
 
An important feature of TOST that distinguishes it from
significance testing is its behaviour under large samples.
With sufficiently many observations, a significance test will
reject $H_0: \beta_A = 0$ for any non-zero gap, flagging models
as unfair when the disparity is trivially small. TOST is
immune to this because it asks whether the gap is small
enough, not whether it is exactly zero. Large samples
help TOST by narrowing the confidence interval, making it
easier to confirm that the gap is within tolerance if the
model is genuinely fair.
 
For PD, the asymmetry of the criterion
means TOST is not applied in the same form as for CDP. The
decision rule instead requires both statistical significance
at the firm-level $\alpha$ and a relative shift exceeding
$\rho_{\min} = 10\%$, as specified in Section~\ref{sec:FairnessHypo}.
Under this two-part rule, we assume a variable is flagged only when
the data provide positive evidence of both statistical and
substantive proxy behaviour; and conversely evidence against substantial impact can be positively generated. 
 
\subsection{Power and sample size}
 
Under TOST, power is the probability that a genuinely fair
model correctly receives a pass verdict. A model receives an
Insufficient Information verdict not only when it is unfair,
but also when the confidence interval is too wide to fall
inside the tolerance band, an outcome that reflects
insufficient data rather than a genuine failure of
compliance. 

As is the case for testing PD and CPD, standard regression sample size calculations cannot be employed when price is deterministic.  However,  tests in both cases are based on a variance that can be estimated from a pilot sample or from historical data on a comparable pricing system, or from quantities obtained from historical audits. 

Both PD and CPD produce estimates $\Delta$ either for an effect estimated in a regression model (CPD) or for the change in coefficents of a regression when a protected attribute covariate is included (PD).  Let $\sigma^2_\Delta$ denote
variance of $\Delta$ obtained from
the pilot study. The power of a test at the assumed true shift
$\Delta^* \ne d$ in which $d$ is the relevant threshold in the TOST framework is
\[
  \pi \approx 1 - \Phi\!\left(
    z_{1-\alpha} -
    \frac{|\Delta^* - d|}{\sigma_\Delta}
  \right).
\]
Observing that the covariances in Propositions \ref{prop:ols}-\ref{prop:shift} all scale with $n$, if there are $n_0$ observations in the pilot data, the requried audit sample to detect a true shift of magnitude $|\Delta^* - d|$ with power $1-\beta$ is 
\[
n \ge \frac{n_0 \sigma^2_{\Delta}(z_{1-\alpha} + z_{1-\beta})^2}{|\Delta^* - d|^2}.
\]
In the case of PD, we test $\Delta^*_{\mathrm{PD},j} \ne 0$ setting $d = 0$ and run a single test. For CDP, we test both $\Delta^*_{\mu} \geq \delta$ and $|\log R_{\mu} | \geq |\log \tau|$. Here the required sample size is computed separately for each test and the larger taken as the planning target.

\subsection{The audit protocol}
 
Figure~\ref{fig:audit-flow} summarises the end-to-end audit flow. All design
choices must be fixed before examining the data. Post-hoc
decisions about which tests to run or which margins to apply
undermine the stated type-I error guarantees and reduce the
auditability of the process.
 
\begin{figure}[h!]
\centering
\scalebox{0.45}{
    \begin{tikzpicture}[
  >=Latex,
  font=\sffamily,
  node distance=12mm and 20mm,
  box/.style={draw, rounded corners=2mm, align=center, inner sep=5mm,
    minimum width=70mm, minimum height=22mm, fill=white},
  decision/.style={draw, diamond, aspect=2, align=center,
    inner sep=3mm, fill=white},
  line/.style={-Latex, line width=.5mm},
  label/.style={font=\bfseries\sffamily, text=textgray},
  note/.style={font=\sffamily, text=textgray}
]
\node[box, fill=plan!15, label={[label, text=plan]above:PLAN}]
  (criterion) {\textbf{Pre-audit Setup}};
\node[box, fill=plan!15, right=of criterion] (design)
  {\textbf{Design Quote Sample}\\[2mm] Representative};
\node[box, fill=audit!15, below=of criterion,
  label={[label, text=audit]above:AUDIT}] (collect)
  {\textbf{Statistical \color{blue}{Fairness Testing}}\\[2mm]
   Estimate effects \& CIs};
\node[box, fill=audit!15, below=of design] (analyze)
  {\textbf{Collect Quotes/Sample}};
\node[decision, fill=decide!15, below=15mm of collect, xshift=35mm,
  label={[label, text=decide]above:DECIDE}] (equiv)
  {\textbf{Within}\\[1mm]\textbf{Tolerance?}};
\node[box, fill=decide!15, below left=15mm and 45mm of equiv,
  minimum width=50mm, minimum height=18mm] (pass)
  {\textbf{Pass}\\[2mm] Within margins};
\node[box, fill=decide!15, below=25mm of equiv,
  minimum width=60mm, minimum height=18mm] (insuff)
  {\textbf{Further Investigation}\\[2mm] CI too wide or on boundary};
\node[box, fill=decide!15, below right=15mm and 45mm of equiv,
  minimum width=50mm, minimum height=18mm] (fail)
  {\textbf{Fail}\\[2mm] Outside margins};
\node[box, fill=improve!15, below=35mm of insuff,
  label={[label, text=improve]above:IMPROVE}] (remed)
  {\textbf{Remediation}\\[2mm] Review \& re-test};
\draw[-Latex, thick, color=plan!80!black]
  ([yshift=15mm]criterion.north) -- (criterion.north)
  node[midway, note, above=8mm of criterion]{START};
\draw[line, color=plan!80!black]   (criterion) -- (design);
\draw[line, color=audit!80!black]  (design) -- (analyze);
\draw[line, color=audit!80!black]  (analyze) -- (collect);
\draw[line, color=decide!80!black] (collect.south) -- (equiv.north);
\draw[line, color=decide!80!black] (equiv) -|
  node[near start, above left, note]{Pass} (pass);
\draw[line, color=decide!80!black] (equiv) --
  node[midway, right, note]{Unclear} (insuff);
\draw[line, color=decide!80!black] (equiv) -|
  node[near start, above right, note]{Fail} (fail);
\draw[line, color=improve!80!black] (fail) -- (remed);
\coordinate (rightboundary) at ([xshift=40mm]design.east);
\draw[-Latex, thick, color=decide!70!black, rounded corners=6mm]
  (insuff.east) -- (rightboundary |- insuff.east)
  -- (rightboundary |- design.east) -- (design.east)
  node[pos=0.75, right, note, xshift=27mm]{Collect more data};
\draw[line, color=decide!70!black]
  (insuff.south) -- (remed.north)
  node[midway, right, note]{Escalate};
\coordinate (leftboundary) at ([xshift=-55mm]criterion.west);
\draw[-Latex, thick, color=improve!70!black, rounded corners=6mm]
  (remed.west) -- (leftboundary |- remed.west)
  -- (leftboundary |- criterion.west) -- (criterion.west)
  node[pos=0.08, left, note, xshift=-2mm]{Re-test};
\end{tikzpicture}}
 
\smallskip
\noindent\hfill
{\color{plan}\rule{5mm}{2.5mm}}~Plan \quad
{\color{audit}\rule{5mm}{2.5mm}}~Audit \quad
{\color{decide}\rule{5mm}{2.5mm}}~Decide \quad
{\color{improve}\rule{5mm}{2.5mm}}~Improve
\hfill\null
 
\caption{End-to-end flow for the fairness audit protocol.}
\label{fig:audit-flow}
\end{figure}
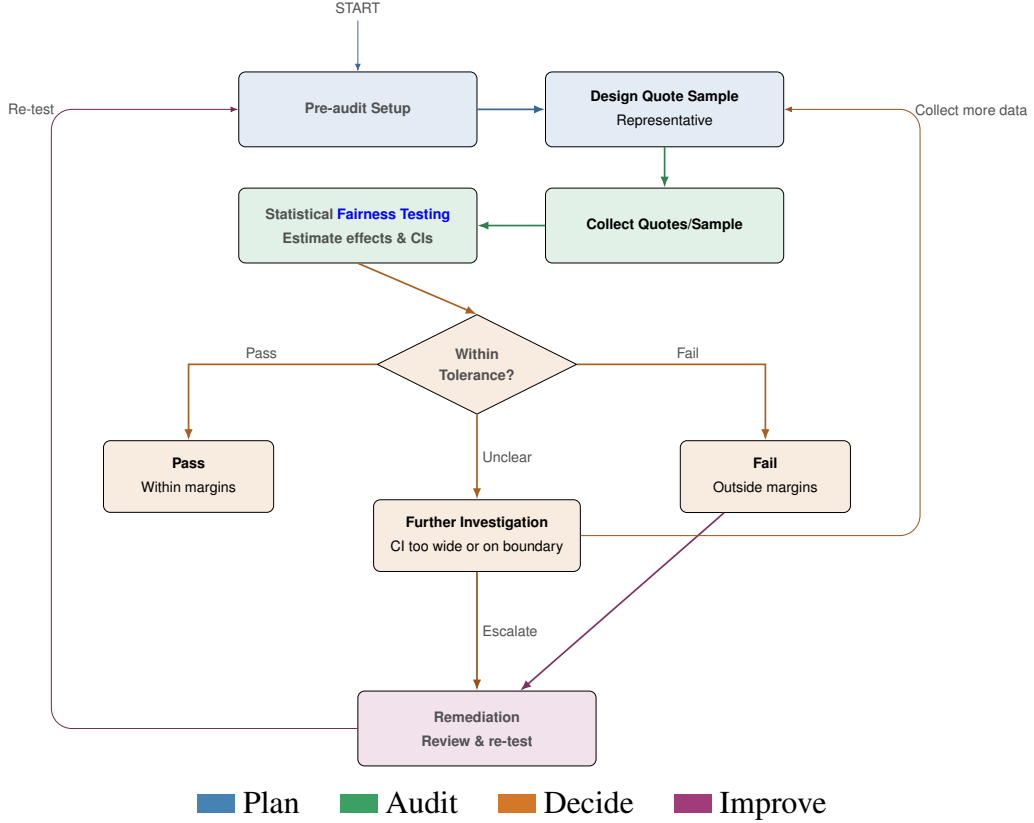
 
\paragraph{Pre-audit setup.}
\begin{enumerate}
 
\item \textbf{Select a fairness criterion} from
  \{PD, CDP\}. PD is appropriate when the audit targets
  specific rating variables suspected of acting as proxies
  for the protected attribute. CDP is appropriate when the
  audit evaluates the system-level price gap after controlling
  for legitimate risk factors.  Other fairness notions may be appropriate depending on the application context and legislative requirements \citep{xin-huang-2023-anti}. 
 
\item \textbf{Specify the protected attribute} $A$. If not
  directly observed, document the inference method and its
  expected accuracy.
 
\item \textbf{Specify legitimate rating factors} $X_\ell$.
  These are the variables the regulator recognises as
  actuarially justified risk predictors. The choice should
  be documented and, where possible, grounded in the
  applicable regulatory guidance.
 
\item \textbf{Define the response variable} $P$.
  Common choices are the quoted premium, pure premium, loss
  ratio, or approval indicator. The choice determines the
  economic interpretation of the disparity and should match
  the regulatory question being addressed.
 
\item \textbf{Set tolerance margins}: level-gap margin
  $\delta$, ratio margin $\tau$, significance level
  $\alpha$, and minimum relative shift $\rho_{\min}$ for PD.
 
\item \textbf{Specify the standard error estimator.}
  Use HC3 by default for CDP. For the PD coefficient-shift
  test, use the full cross-covariance formula of
  Proposition~\ref{prop:shift}. Both apply because pricing
  outputs are deterministic (Propositions~\ref{prop:ols}
  and~\ref{prop:glm}).
 
\item \textbf{Define test scope.}
  Each company is audited independently at the pre-specified
  $\alpha$. No cross-company correction is applied. If a
  single company's audit involves multiple proxy variables,
  within-company correction across those variables is
  appropriate.
 
\item \textbf{Design the quote sample.} Construct a
  representative sample of audit profiles from the insured
  portfolio, stratified to reflect the mix of risk
  characteristics in the covered population. Both PD and
  CDP use a representative sample. The sample size should
  satisfy the power requirement from Section~\ref{sec:framework}.
 
\end{enumerate}

All quote requests should be submitted within a short time
window to minimise the risk of model updates between
submissions. For each observation, record $(P, A, X_\ell,
t, \text{model version})$ and store the complete dataset
along with a hashed manifest for reproducibility. Apply the
criterion-specific test using the pre-specified standard
error estimator and decision rule.

The final verdict is one of three outcomes. Pass:
the confidence interval lies entirely within the tolerance
band, providing positive evidence of fairness. Fail:
the confidence interval lies entirely outside the tolerance
band, establishing a material disparity. Insufficient
information: the confidence interval is too wide to determine
whether the true disparity lies within or outside the tolerance
band. This outcome does not constitute a pass and should
prompt additional data collection or a pre-specified
escalation procedure. When a model fails, the proxy
discrimination screen runs the PD test to identify candidate
driver variables, heterogeneous effects analysis re-runs the
failing test within risk tiers and geographic segments, and
remediation options include removing or recalibrating the
driver variable, imposing fairness constraints during model
retraining, or implementing ongoing monitoring with a trigger
threshold for re-audit.
 
\section{Empirical Validation}
\label{sec:empirical}
 
\subsection{Data and setting}
 
We use the ProPublica Illinois auto insurance dataset
\citep{larson-etal-2017-minority}, which comprises 31,382
zip-code-company observations for 34 insurers across 923
zip codes. The response is the annual combined liability
premium quoted
for a standardised driver profile, a 30-year-old female
schoolteacher with excellent credit, no accidents, driving
a 2016 Toyota Camry. These quotes are deterministic outputs. The same profile
submitted to the same company at the same point in time
always yields the same premium, the property that
necessitates the inferential corrections developed in
Section~\ref{sec:inference}.
 
The protected attribute is the zip-code minority flag ($A = 1$
for zip codes with at least 50\% non-white population). Of the
923 zip codes, 104 (11.3\%) are classified as majority-minority.
The audit model uses log state risk (the Illinois DOI aggregate
loss cost per insured vehicle) and a Chicago indicator as the
legitimate rating factors $X_\ell$.
 
The insurance industry contested the ProPublica analysis,
arguing that the aggregate loss cost variable is an imperfect
proxy for each insurer's individual risk model
\citep{iii-2017-propublica}. We use this dataset as a
well-documented illustration of the audit framework, not
as the basis for conclusions about discriminatory intent.
Table~\ref{tab:data-summary} provides summary statistics.
Majority-minority zip codes pay mean premiums 35.5\% higher
than white zip codes (\$482 vs.\ \$356), while their state
risk is only 11.6\% higher. The regression analysis below
controls for these factors.
 
\begin{table}[h]
\centering
\caption{Summary statistics and unconditional comparisons,
  Illinois auto insurance dataset ($n = 31{,}382$ zip-company
  observations).}
\label{tab:data-summary}
\begin{tabular}{lrrrrrr}
\hline
& \multicolumn{3}{c}{\textbf{Full sample}} &
  \multicolumn{3}{c}{\textbf{By minority status}} \\
\cline{2-4}\cline{5-7}
\textbf{Variable} & Mean & SD & Max &
  White zips & Min.\ zips & Ratio \\
\hline
Combined premium (\$) & 370 & 148 & 1345 & 356 & 482 & 1.355 \\
State risk (\$)       & 163 &  51 &  298 & 161 & 179 & 1.116 \\
Excess premium (\$)   & 207 & 108 &  --- & 195 & 303 & 1.551 \\
Pct minority (\%)     &  19 &  24 &   99 & -- & -- & --   \\
\hline
\end{tabular}
\end{table}
 
\subsection{Measurement validity: the cost of incorrect standard errors}
 
For each of the 34 companies we fit the CDP regression by
OLS and compare two standard error estimators, the classical
formula $\hat\sigma^2(X^\top X)^{-1}$ and the HC3 sandwich
estimator. Table~\ref{tab:se-comparison} reports the ratio
$\rho_k = \mathrm{SE}_{\mathrm{HC3}} / \mathrm{SE}_{\mathrm{classical}}$
for all 34 companies.
 
The ratios range from 0.685 to 1.775, with a mean of 1.065.
For 14 companies, the departure from 1.0 exceeds 0.15,
large enough to change test outcomes. The classical formula understates the true standard error
for some companies and overstates it for others. The
direction cannot be determined without computing both.
 
\begin{table}[h]
\centering
\caption{CDP regression: ratio of HC3 to classical standard
  errors for the minority coefficient $\hat\beta_{A,k}$, selected
  companies. Full results in Table~\ref{tab:se-full} (Appendix~\ref{app:empirical}).
  $\rho = \mathrm{SE}_{\mathrm{HC3}} / \mathrm{SE}_{\mathrm{classical}}$.}
\label{tab:se-comparison}
\begin{tabular}{lrrrrr}
\hline
\textbf{Company} & $\hat\beta_A$ & $\mathrm{SE}_{\mathrm{cl}}$ &
  $\mathrm{SE}_{\mathrm{HC3}}$ & $\rho$ & $R^2$ \\
\hline
Economy Preferred Ins Co    & +0.308 & 0.017 & 0.030 & 1.775 & 0.476 \\
Metropolitan Cas Ins Co     & +0.308 & 0.017 & 0.029 & 1.679 & 0.482 \\
Farmers Automobile Ins Assoc& +0.294 & 0.024 & 0.040 & 1.660 & 0.675 \\
\multicolumn{6}{c}{\emph{\ldots\ 28 companies ($\rho$ ranging 0.80--1.33) \ldots}} \\
State Farm Mut Auto Ins Co  & +0.222 & 0.017 & 0.013 & 0.753 & 0.450 \\
Trumbull Ins Co             & +0.130 & 0.013 & 0.009 & 0.685 & 0.405 \\
\hline
\multicolumn{3}{l}{Mean $\rho$ = 1.065 \quad Median $\rho$ = 1.051} & & & \\
\hline
\end{tabular}
\end{table}
 
The $R^2$ values (mean 43.6\%) indicate that the linear
audit model leaves substantial pricing variation unexplained,
reflecting territory relativities, tiering structures, and
other components not captured by the three-variable audit
model. These unexplained components create the approximation
errors that drive the divergence between classical and sandwich
standard errors.
 
The practical implication is significant. No universal
conservative correction is available, since the direction of
the discrepancy varies across companies in the same dataset
and the same audit. Practitioners who rely on standard OLS
output will sometimes under-reject, missing real disparities,
and sometimes over-reject, flagging non-existent ones.
 
\subsection{Full audit under the corrected framework: PD}
 
We apply the full PD audit to the log state risk coefficient.
This is the variable most likely to serve as a proxy for
race given the correlation between zip-code risk levels and
minority population share in Illinois, and it is the only
rating variable available in the publicly accessible dataset.
A richer regulatory dataset would permit testing of multiple
proxy candidates. In that setting, within-company multiplicity
correction across the tested variables would be appropriate,
as specified in step 7 of the pre-audit setup. The pre-audit setup specifies
criterion (PD), proxy variable (log state risk), protected
attribute (minority flag), significance level $\alpha = 0.05$,
minimum relative shift $\rho_{\min} = 10\%$, and standard error
estimator (score-based sandwich, Proposition~\ref{prop:shift}).
Each company is tested independently at $\alpha = 0.05$.

We compare the log-risk coefficient across the restricted
and extended models for each company, computing the
standard error of the shift under both the independent-samples
formula and the cross-covariance corrected formula of
Proposition~\ref{prop:shift}. The correction reduces the
standard error in every case. The ratio
$\mathrm{SE}_{\mathrm{full}} / \mathrm{SE}_{\mathrm{ind}}$
has mean 0.082, so the corrected standard error is on
average 92\% smaller. Under the independent-samples
formula, zero companies produce $|z| > 1.645$. Under
the corrected formula, all 34 do. Three examples illustrate the mechanism.
 
For Metropolitan Prop \& Cas Ins Co, the log-risk coefficient
shifts by 12.5\% when the minority flag is added. The
independent-samples standard error is 0.0245, giving
$z_{\mathrm{ind}} = 1.34$, which is not significant. The
corrected standard error is 0.0021, a reduction by a factor
of 11.7, and $z_{\mathrm{full}} = 15.73$ is significant at
the 0.1\% level.
 
For Farmers Automobile Ins Assoc, a shift of 22.2\% is
insignificant under the independent-samples formula
($z_{\mathrm{ind}} = 1.07$) but highly significant under the
corrected formula ($z_{\mathrm{full}} = 12.59$).
 
For Geico Ind Co, the ratio is 0.081 and the corrected
formula gives $z_{\mathrm{full}} = 7.09$ compared to
$z_{\mathrm{ind}} = 0.58$ under the independent-samples
formula, a company that appeared comfortably
non-significant is revealed to be highly significant once
the cross-covariance correction is applied.

Applying the two-part decision rule at the firm level, 
$|z_{\mathrm{full}}| > 1.645$ and relative shift $> 10\%$, flags 16 of the 34 companies as proxy discriminators.
Each company is audited independently. No cross-company
correction is applied, as each verdict concerns only the
firm being tested. The 16 flagged companies are Allstate
Indemnity, Country Mutual, Country Preferred, Economy
Preferred, Farmers Automobile, Geico General, Geico
Indemnity, Government Employees, Metropolitan Casualty,
Metropolitan Prop \& Cas, Owners, Safeco, Travelers
Commercial, Travelers Home \& Marine, United Services
Automobile, and USAA Casualty, with relative shifts ranging
from 10.2\% to 22.2\%. All 34 companies are individually
significant under the corrected formula. The remaining 18
do not meet the 10\% substantive shift threshold and are
not flagged. Full results are in
Table~\ref{tab:proxy-full} (Appendix~\ref{app:empirical:proxy}). Without
the cross-covariance correction, none of the 34 companies
would be flagged under either criterion.

\begin{table}[h]
\centering
\caption{PD audit results: selected companies illustrating the
  inferential correction. Shift $= |(\hat\phi -
  \hat\phi')/\hat\phi|$, the relative shift in the log-risk
  coefficient when the minority flag is added.
  $\mathrm{SE}_{\mathrm{ind}}$ uses the independent-samples
  formula; $\mathrm{SE}_{\mathrm{full}}$ uses the corrected
  cross-covariance formula of Proposition~\ref{prop:shift}.
  Dec.\ $=$ two-part verdict ($|z_{\mathrm{full}}| > 1.645$
  and Shift $> 10\%$). Full results in
  Table~\ref{tab:proxy-full}.}
\label{tab:pd-results}
\begin{tabular}{lrrrrrl}
\hline
\textbf{Company} & \textbf{Shift} &
  $\mathrm{SE}_{\mathrm{ind}}$ &
  $\mathrm{SE}_{\mathrm{full}}$ &
  $z_{\mathrm{ind}}$ & $z_{\mathrm{full}}$ &
  \textbf{Dec.} \\
\hline
\multicolumn{7}{l}{\textit{Panel A: flagged companies (selected)}} \\
Farmers Automobile Ins Assoc    & 22.2\% & 0.0253 & 0.0022 &  1.07 & 12.59 & FLAG \\
Economy Preferred Ins Co        & 16.5\% & 0.0248 & 0.0021 &  1.15 & 13.71 & FLAG \\
Metropolitan Prop \& Cas Ins Co & 12.5\% & 0.0245 & 0.0021 &  1.34 & 15.73 & FLAG \\
Geico Ind Co                    & 10.2\% & 0.0276 & 0.0022 &  0.58 &  7.09 & FLAG \\[3pt]
\multicolumn{7}{l}{\textit{Panel B: significant but below threshold (selected)}} \\
Metropolitan Grp Prop \& Cas    &  9.3\% & 0.0268 & 0.0022 &  0.96 & 11.59 & --- \\
Garrison Prop \& Cas Ins Co     &  7.8\% & 0.0250 & 0.0020 &  0.32 &  4.00 & --- \\
\hline
\multicolumn{7}{l}{Mean SE ratio (full/ind) $= 0.082$;\enspace
  $|z_{\mathrm{ind}}| > 1.645$: 0 of 34;\enspace
  $|z_{\mathrm{full}}| > 1.645$: 34 of 34} \\
\hline
\end{tabular}
\end{table}

 Table~\ref{tab:pd-results} illustrates the mechanism for
selected companies. Panel A shows four flagged companies
spanning the range of relative shifts. Panel B shows two
companies that are statistically significant under the
corrected formula but do not meet the 10\% substantive
threshold. The pattern is uniform across all 34 companies. The ratio
$\mathrm{SE}_{\mathrm{full}}/\mathrm{SE}_{\mathrm{ind}}$
has mean 0.082 and range $[0.080, 0.085]$, so the
correction is of similar magnitude for every company in
the dataset regardless of pricing structure.

\subsection{Full audit under the corrected framework: CDP}
 
We apply the full audit protocol to the CDP criterion.
The pre-audit setup specifies criterion (CDP), protected
attribute (minority flag), rating factors (log state risk
and Chicago indicator), response variable (log combined
premium), tolerance margins ($\delta = 5\%$ of mean premium
$= \$18.51$, $\tau = 0.80$, $\alpha = 0.05$), standard error
estimator (HC3). Each company is audited independently
at $\alpha = 0.05$.

 We use the log-linear specification of
Equation~\ref{eq:cdp} for three reasons. Insurance
premiums are right-skewed and strictly positive, making
the log transformation the natural variance-stabilising
choice. The multiplicative pricing structure of insurance
rating systems means the ratio $R_\mu$ is the more
interpretable fairness measure, and under the log-linear
model $\hat\beta$ estimates $\log R_\mu$ directly.
Finally, the Colorado DOI tolerance of $\tau = 0.80$
is naturally expressed as $|\hat\beta| < \log(1/0.80)
= 0.223$, a bound directly on the estimated coefficient.

For each company $k$ we estimate a company-specific
version of Equation~\ref{eq:cdp}, where observations
are indexed by company $k$ and zip code $z$, and the
legitimate rating factors $X_\ell$ consist of two
variables: log state risk and a Chicago indicator.
The estimated model is
\begin{equation}
\log P_{kz} = \mu_{0k} + \beta_{A,k} \cdot \mathbf{1}\{A_z = 1\}
  + \gamma_k \log(\text{StateRisk}_z)
  + \psi_k \cdot \text{Chicago}_z + r_{kz},
\label{eq:cdp-case}
\end{equation}
where $\mu_{0k}$ is the company-specific intercept, $\gamma_k$ is the coefficient on log state risk, $\psi_k$ is the coefficient on the Chicago indicator, $\beta_{A,k}$ captures the conditional log-premium 
gap for company $k$ alone, and $r_{kz}$ is the approximation error for company $k$
at zip code $z$. Each
company is estimated separately so that $\hat\beta_{A,k}$
captures the conditional log-premium gap for company $k$
alone. A company
passes CDP if and only if the 90\% CI for $\beta_{A,k}$ lies
entirely within
$(\log 0.80, -\log 0.80) = (-0.223, +0.223)$
and the implied dollar gap $\bar{P}(\exp(\hat\beta_{A,k}-1)$
lies within $(\pm\$18.51)$.
 
\begin{table}[h]
\centering
\caption{CDP audit results: companies at the top and bottom
  of the disparity distribution. Full results in Table~\ref{tab:cdp-full} (Appendix~\ref{app:empirical}). $\beta_A$ = conditional log-premium gap
  (HC3 SEs); Ratio $= e^{\beta_A}$; Gap = dollar difference at
  mean premium (\$370). TOST: $\delta = 5\%$ of mean,
  $\tau = 0.80$, $\alpha = 0.05$.}
\label{tab:cdp-results}
\begin{tabular}{lrrrr}
\hline
\textbf{Company} & \textbf{Gap (\$)} & \textbf{Ratio} &
  \textbf{90\% CI for $\beta_A$} & \textbf{Dec.} \\
\hline
Metropolitan Prop \& Cas Ins Co & \$158 & 1.427 &
  $[+0.315,\;+0.397]$ & FAIL \\
Allstate Ind Co                 & \$138 & 1.374 &
  $[+0.287,\;+0.349]$ & FAIL \\
\multicolumn{5}{c}{\emph{\ldots\ 29 companies
  (all FAIL, ratios 1.10--1.37) \ldots}} \\
USAA Cas Ins Co                 & \$ 35 & 1.095 &
  $[+0.079,\;+0.102]$ & FAIL \\
Garrison Prop \& Cas Ins Co     & \$ 34 & 1.091 &
  $[+0.074,\;+0.100]$ & FAIL \\
\hline
\end{tabular}
\end{table}
 
All 34 companies fail the CDP test. Price ratios range from
1.09 to 1.43, implying annual premiums in majority-minority
zip codes \$34 to \$158 higher than in comparable-risk white
zip codes. None of the 90\% confidence intervals approaches
the upper tolerance boundary of $+0.223$. For comparison, all 34
companies are also flagged under conventional significance
testing ($H_0: \beta_k = 0$, $\alpha = 0.05$) because the
disparities are large enough to be detected under either
inferential framework. The added value of TOST for CDP is
therefore not a difference in detection but in
interpretation. TOST provides a positive compliance verdict
only when the confidence interval falls entirely within the
tolerance region, a standard that no insurer meets here,
rather than simply failing to reject a null of zero
disparity.
 
These results confirm the core finding of
\citet{larson-etal-2017-minority}. The TOST framing adds two things the original analysis
could not provide: a formal pass-or-fail verdict tied to a
pre-specified tolerance, and confidence intervals that
quantify
estimation precision rather than simply detecting a
non-zero gap. The value of the equivalence framing is
greatest in datasets with more modest disparities, where
significance testing flags gaps that fall within
regulatory tolerance.
 
\subsection{Limitations}
 
Several limitations bear on the interpretation of these
results. The analysis operates at the zip code level, so
disparities reflect neighbourhood-level averages rather
than individual pricing decisions. The minority flag is
derived from census composition rather than directly
observed. \citet{kallus-etal-2021-assessing} show that
proxy-based measurement can render disparity estimates
unidentifiable. The audit model controls for two variables,
and a disparity that survives those controls may be
partially explained by approved factors we cannot observe,
including territory relativities and tier assignments. The
state risk variable is a 2012--2014 aggregate applied to
2017 premiums and does not capture each insurer's
individual loss experience. The data are standardised
quotes for a single driver profile, not actual premiums
paid. These limitations are inherent to the publicly
available dataset and do not affect the validity of the
inferential corrections, which apply equally to the richer
data available to regulators.

\section{Discussion}
\label{sec:discussion}
 
\subsection{Implications for bias detection in algorithmic systems}

The IS literature on algorithmic bias
\citep{fu-etal-2021-crowds, hu-etal-2025-human} uses
hypothesis testing to detect bias. When the
response is stochastic, this is a reasonable approximation.
When the response is a deterministic algorithm output, the
independence assumption fails by construction and tests
can be invalid. Any IS study using statistical testing in a deterministic
algorithmic context should revisit its null findings
with the corrected variance.

The inferential problem we identify is not limited to
fairness auditing. Any IS study that regresses a
deterministic algorithmic output on individual or group
characteristics, such as personalisation effects, pricing
discrimination, recommendation diversity, faces the
same structural mismatch between the data-generating
process and the inference machinery. The output is a
fixed function of inputs, not a draw from a distribution,
so classical standard errors produce uncertainty estimates
that are wrong in both direction and magnitude.
Propositions~\ref{prop:ols} and~\ref{prop:glm} establish
the correct procedure and recommend using HC3 by default for any regression
whose response is an algorithmic output. The computational
cost is negligible, and the inferential cost of not doing so
is, as the Illinois data show, material.

\subsection{The burden of proof in algorithmic fairness}
 
TOST carries accountability consequences that extend beyond
the statistical properties of the test. Under significance
testing, the algorithm is presumed fair unless the data
establishes otherwise. The regulator bears the evidentiary
burden. TOST reverses this allocation. The firm bears the burden
of demonstrating that its pricing gap falls within the
regulatory tolerance, which is the appropriate structure
when the firm controls the system being tested.
 
The reversal matters most in large administrative datasets.
With the sample sizes typical of insurance, credit, and
mortgage audits, a significance test detects any non-zero
disparity, including those within regulatory tolerance.
Firms with genuinely fair systems cannot obtain a clean
bill of health under significance testing, because that
framework can only fail to reject the null of unfairness,
not affirm compliance. TOST provides this affirmation.
It also produces a third verdict that significance testing
cannot. When the confidence interval is too wide to fall
entirely within or outside the tolerance band, the
appropriate verdict is Insufficient Information rather
than an implicit pass, directing additional data collection
before the audit closes.

\subsection{Consumer harm implications}
 
The inferential errors translate into identifiable 
costs. When an audit incorrectly clears an insurer whose
pricing violates conditional demographic parity, minority
policyholders continue to pay excess premiums without
the regulatory correction that a valid audit would trigger.
In the Illinois data, those excess premiums range from \$34
to \$158 per year. The CDP disparities are large enough
that they would have been detected under either classical
or HC3 standard errors. The CDP inference problem is
mainly about precision, not detection. The proxy
discrimination problem is more consequential.
Under the standard variance formula, zero companies were
flagged. Under the correct formula, 17 were identified. Rating variables that proxy for race, through occupation, zip code, or credit attributes, continue to operate undisturbed when the test lacks the power to detect them. The cross-covariance correction restores that power and supports regulatory requirements that firms justify or remove such variables.

\section{Conclusion}
\label{sec:conclusion}
 
The regression-based tests used to audit algorithmic
pricing systems are structurally misspecified for the
deterministic outputs they test. We derive the correct
sandwich variance for OLS and GLM audit regressions and
the correct cross-covariance formula for the
coefficient-shift test used to detect proxy
discrimination. Beyond the inferential problem, we
formalise proxy discrimination and conditional
demographic parity as named testable criteria, adopt
equivalence testing as the decision framework so that
a compliance verdict requires affirmative evidence
rather than a failure to detect, and combine these
elements into a complete pre-specified audit protocol.
 
The empirical results underscore the practical stakes.
Every Illinois insurer fails the conditional demographic
parity test, with minority zip codes paying \$34--\$158
more per year than comparable-risk white zip codes after
controlling for risk and geography. The proxy
discrimination results are more striking still: the
cross-covariance correction reduces the standard error
by 92 percent on average, reclassifying all 34 insurers
from non-significant to statistically significant.
Sixteen meet both the significance and substantive shift
thresholds and are formally flagged. None would be
flagged under standard practice. Together these findings
establish that the inferential errors in current audit
practice have significant policy and practice impacts. 
 
Three directions for future work follow directly. The
tolerance margins used here are taken from existing
regulatory proposals. A welfare-theoretic basis for
calibrating them, connecting the margins to the consumer
cost of excess premiums, would strengthen the normative
foundations of the framework. The analysis treats each
audit as a cross-sectional exercise at a single point in
time. Extending the inferential results to panel settings
and longitudinal monitoring would increase practical
relevance. And while the Illinois auto insurance data
provide a well-documented validation context, applying
the framework to credit scoring, mortgage pricing, and
hiring algorithms would establish the generality of the
inferential corrections and the audit protocol across
the full range of consequential deterministic systems.

 
 
\phantomsection
\addcontentsline{toc}{section}{References}
\bibliographystyle{apacite}
\bibliography{references_ms}

@misc{iii-2017-propublica,
  author       = {Lynch, James},
  title        = {Why {ProPublica}'s Auto Insurance Report Is
                  Inaccurate, Unfair and Irresponsible},
  howpublished = {Insurance Information Institute},
  month        = apr,
  year         = {2017},
  url          = {https://www.iii.org/article/why-propublicas-auto-insurance-report-is-inaccurate-unfair-and-irresponsible}
}

@article{voicu2018using,
  title={Using first name information to improve race and ethnicity classification},
  author={Voicu, Ioan},
  journal={Statistics and Public Policy},
  volume={5},
  number={1},
  pages={1--13},
  year={2018},
  publisher={Taylor \& Francis}
}

@misc{xin2026,
      title={How Proxy Race Distorts Regression-Based Fairness Audits}, 
      author={Xi Xin and Giles Hooker and Fei Huang},
      year={2026},
      eprint={2603.17106},
      archivePrefix={arXiv},
      primaryClass={stat.AP},
      url={https://arxiv.org/abs/2603.17106}, 
}

@article{schuirmann-1987-comparison,
  author  = {Schuirmann, Donald J},
  title   = {A Comparison of the Two One-Sided Tests Procedure and the
             Power Approach for Assessing the Equivalence of Average
             Bioavailability},
  journal = {Journal of Pharmacokinetics and Biopharmaceutics},
  volume  = {15},
  number  = {6},
  pages   = {657--680},
  year    = {1987}
}

@article{mackinnon-white-1985,
  author  = {MacKinnon, James G and White, Halbert},
  title   = {Some Heteroskedasticity-Consistent Covariance Matrix Estimators
             with Improved Finite Sample Properties},
  journal = {Journal of Econometrics},
  volume  = {29},
  number  = {3},
  pages   = {305--325},
  year    = {1985}
}

@article{mehrabi-etal-2021-survey,
  title={A survey on bias and fairness in machine learning},
  author={Mehrabi, Ninareh and Morstatter, Fred and Saxena, Nripsuta and Lerman, Kristina and Galstyan, Aram},
  journal={ACM computing surveys (CSUR)},
  volume={54},
  number={6},
  pages={1--35},
  year={2021},
  publisher={ACM New York, NY, USA}
}

@book{barocas-etal-2022-fairness,
  title = {Fairness and Machine Learning: Limitations and Opportunities},
  author = {Solon Barocas and Moritz Hardt and Arvind Narayanan},
  publisher = {MIT Press},
    url = {https://fairmlbook.org/},
  year = {2023}
}

@article{xin-huang-2023-anti,
  title={Antidiscrimination insurance pricing: Regulations, fairness criteria, and models},
  author={Xin, Xi and Huang, Fei},
  journal={North American Actuarial Journal},
  pages={1--35},
  year={2023},
  publisher={Taylor \& Francis}
}

@article{lindholm-etal-2022-discrimination,
  title={Discrimination-free insurance pricing},
  author={Lindholm, Mathias and Richman, Ronald and Tsanakas, Andreas and W{\"u}thrich, Mario V},
  journal={ASTIN Bulletin: The Journal of the IAA},
  volume={52},
  number={1},
  pages={55--89},
  year={2022},
  publisher={Cambridge University Press}
}

@article{prince-schwarcz-2020-proxy,
  title={Proxy discrimination in the age of artificial intelligence and big data},
  author={Prince, Anya ER and Schwarcz, Daniel},
  journal={Iowa L. Rev.},
  volume={105},
  pages={1257},
  year={2020}
}

@misc{doradoi-2023-concerning,
author = {{Colorado Division of Insurance}},
title = {CONCERNING QUANTITATIVE TESTING OF EXTERNAL CONSUMER DATA AND
INFORMATION SOURCES, ALGORITHMS, AND PREDICTIVE MODELS USED FOR
LIFE INSURANCE UNDERWRITING FOR UNFAIRLY DISCRIMINATORY
OUTCOMES},
howpublished = {\url{https://drive.google.com/file/d/1BMFuRKbh39Q7YckPqrhrCRuWp29vJ44O/view
}},
note = {Accessed: 12 June 2024},
year = {2023}
}

@article{duPreez-etal-2024-bias,
  title={From bias to black boxes: understanding and managing the risks of AI--an actuarial perspective},
  author={du Preez, Valerie and Bennet, Shaun and Byrne, Matthew and Couloumy, Aureli{\'e}n and Das, Arijit and Dessain, Jean and Galbraith, Richard and King, Paul and Mutanga, Victor and Schiller, Frank and others},
  journal={British Actuarial Journal},
  volume={29},
  pages={e6},
  year={2024},
  publisher={Cambridge University Press}
}

@article{pope-sydnor-2011-implementing,
author = {Pope, Devin G. and Sydnor, Justin R.},
copyright = {Copyright © 2011 American Economic Association},
issn = {1945-7731},
journal = {American Economic Journal: Economic Policy},
language = {eng},
number = {3},
pages = {206-231},
publisher = {American Economic Association},
title = {Implementing Anti-Discrimination Policies in Statistical Profiling Models},
volume = {3},
year = {2011},
}

@misc{ftc-2007-creditBased,
author = {{Federal Trade Commission}},
title = {CREDIT-BASED INSURANCE SCORES:
IMPACTS ON CONSUMERS
OF AUTOMOBILE INSURANCE},
howpublished = {\url{https://www.ftc.gov/sites/default/files/documents/reports/credit-based-insurance-scores-impacts-consumers-automobile-insurance-report-congress-federal-trade/p044804facta_report_credit-based_insurance_scores.pdf}},
note = {Accessed: 24 June 2024},
year = {2007}
}

@misc{nydfs-2024-proposed,
author = {{New York State Department of Financial Services}},
title = {Proposed Insurance Circular Letter: January 17, 2024},
howpublished = {\url{https://www.dfs.ny.gov/industry_guidance/circular_letters/cl2024_nn_proposed}},
note = {Accessed: 29 June 2024},
year = {2024}
}

@misc{larson-etal-2017-minority,
  author       = {Larson, Jeff and Angwin, Julia and Kirchner, Lauren and Mattu, Surya},
  title        = {Minority Neighborhoods Pay Higher Car Insurance Premiums Than White
                  Areas With the Same Risk},
  howpublished = {ProPublica and Consumer Reports},
  year         = {2017},
  url          = {https://www.propublica.org/article/minority-neighborhoods-higher-car-insurance-premiums-white-areas-same-risk},
  note         = {Accessed: April 2026}
}

@article{barocas-moritz-2016-big,
  author  = {Barocas, Solon and Selbst, Andrew D.},
  title   = {Big Data's Disparate Impact},
  journal = {California Law Review},
  volume  = {104},
  number  = {3},
  pages   = {671--732},
  year    = {2016}
}

@article{white-1980-heteroskedasticity,
  author  = {White, Halbert},
  title   = {A Heteroskedasticity-Consistent Covariance Matrix
             Estimator and a Direct Test for Heteroskedasticity},
  journal = {Econometrica},
  volume  = {48},
  number  = {4},
  pages   = {817--838},
  year    = {1980}
}

@incollection{eicker-1967-limit,
  author  = {Eicker, Friedhelm},
  title   = {Limit Theorems for Regressions with Unequal and
             Dependent Errors},
  booktitle = {Proceedings of the Fifth Berkeley Symposium on
              Mathematical Statistics and Probability},
  volume  = {1},
  pages   = {59--82},
  year    = {1967},
  publisher = {University of California Press}
}

@article{fu-etal-2021-crowds,
  author  = {Fu, Runshan and Huang, Yan and Singh, Param Vir},
  title   = {Crowds, Lending, Machine, and Bias},
  journal = {Information Systems Research},
  volume  = {32},
  number  = {1},
  pages   = {72--92},
  year    = {2021}
}

@article{hu-etal-2025-human,
  author  = {Hu, Xiyang and Huang, Yan and Li, Beibei and Lu, Tian},
  title   = {Human--Algorithmic Bias: Source, Evolution, and Impact},
  journal = {Management Science},
  volume  = {72},
  number  = {1},
  pages   = {495--514},
  year    = {2025}
}

@article{zhang-xu-2024-fairness,
  author  = {Zhang, Nan and Xu, Heng},
  title   = {Fairness of Ratemaking for Catastrophe Insurance:
             Lessons from Machine Learning},
  journal = {Information Systems Research},
  volume  = {35},
  number  = {2},
  pages   = {469--488},
  year    = {2024}
}

@article{hurlin-etal-2025-fairness,
  author  = {Hurlin, Christophe and P\'{e}rignon, Christophe
             and Saurin, S\'{e}bastien},
  title   = {The Fairness of Credit Scoring Models},
  journal = {Management Science},
  volume  = {72},
  number  = {1},
  pages   = {406--425},
  year    = {2026}
}

@article{lambrecht-tucker-2019-algorithmic,
  author  = {Lambrecht, Anja and Tucker, Catherine E.},
  title   = {Algorithmic Bias? An Empirical Study of Apparent
             Gender-Based Discrimination in the Display of {STEM}
             Career Ads},
  journal = {Management Science},
  volume  = {65},
  number  = {7},
  pages   = {2966--2981},
  year    = {2019}
}

@article{rhue-2023-anchoring,
  author  = {Rhue, Lauren},
  title   = {The Anchoring Effect, Algorithmic Fairness, and
             the Limits of Information Transparency for Emotion
             Artificial Intelligence},
  journal = {Information Systems Research},
  volume  = {35},
  number  = {3},
  pages   = {1479--1496},
  year    = {2024}
}

@article{pagan-1984-econometric,
  author  = {Pagan, Adrian},
  title   = {Econometric Issues in the Analysis of Regressions
             with Generated Regressors},
  journal = {International Economic Review},
  volume  = {25},
  number  = {1},
  pages   = {221--247},
  year    = {1984}
}

@article{murphy-topel-1985-estimation,
  author  = {Murphy, Kevin M. and Topel, Robert H.},
  title   = {Estimation and Inference in Two-Step Econometric Models},
  journal = {Journal of Business \& Economic Statistics},
  volume  = {3},
  number  = {4},
  pages   = {370--379},
  year    = {1985}
}

@article{fuster-etal-2022-predictably,
  author  = {Fuster, Andreas and Goldsmith-Pinkham, Paul and
             Ramadorai, Tarun and Walther, Ansgar},
  title   = {Predictably Unequal? The Effects of Machine
             Learning on Credit Markets},
  journal = {Journal of Finance},
  volume  = {77},
  number  = {1},
  pages   = {5--47},
  year    = {2022}
}

@article{white-1982-maximum,
  author  = {White, Halbert},
  title   = {Maximum Likelihood Estimation of Misspecified Models},
  journal = {Econometrica},
  volume  = {50},
  number  = {1},
  pages   = {1--25},
  year    = {1982}
}

@article{kallus-etal-2021-assessing,
  author  = {Kallus, Nathan and Mao, Xiaojie and Zhou, Angela},
  title   = {Assessing Algorithmic Fairness with Unobserved
             Protected Class Using Data Combination},
  journal = {Management Science},
  volume  = {68},
  number  = {3},
  pages   = {1959--1981},
  year    = {2021}
}

@article{hausman-1978-specification,
  author  = {Hausman, Jerry A.},
  title   = {Specification Tests in Econometrics},
  journal = {Econometrica},
  volume  = {46},
  number  = {6},
  pages   = {1251--1271},
  year    = {1978}
}

@article{shimao-huang-2024-welfare,
  author  = {Huang, Fei  and  Shimao, Hajime},
  title   = {Welfare Implications of Fair and Accountable Insurance Pricing},
  journal = {SSRN Working Paper},
  year    = {2025},
  note    = {Available at SSRN: \url{https://ssrn.com/abstract=4225159}}
}

@article{shimao-huang-2025-welfare,
  author  = {Huang, Fei  and Shimao, Hajime  and Khern-am-nuai, Warut},
  title   = {Do Fair Algorithms Improve Welfare? Evidence from the Insurance Market},
  journal = {SSRN Working Paper},
  year    = {2025},
  note    = {Available at SSRN: \url{https://ssrn.com/abstract=5112616}}
}

@article{lindholm2022,
  author  = {Lindholm, Mathias and Richman, Ronald and Tsanakas, Andreas
             and W\"{u}thrich, Mario V.},
  title   = {Discrimination-Free Insurance Pricing},
  journal = {ASTIN Bulletin},
  volume  = {52},
  number  = {1},
  pages   = {55--89},
  year    = {2022},
  doi     = {10.1017/asb.2021.23}
}

@article{Frees2023,
  author  = {Frees, Edward W. and Huang, Fei},
  title   = {The Discriminating (Pricing) Actuary},
  journal = {North American Actuarial Journal},
  volume  = {27},
  number  = {1},
  pages   = {2--24},
  year    = {2023},
  doi     = {10.1080/10920277.2021.1951296}
}

@article{xin2022,
  author  = {Xin, Xi and Huang, Fei},
  title   = {Antidiscrimination Insurance Pricing: Regulations,
             Fairness Criteria, and Models},
  journal = {North American Actuarial Journal},
  volume  = {28},
  number  = {2},
  pages   = {285--319},
  year    = {2024},
  doi     = {10.1080/10920277.2023.2190528}
}

@article{vincent2022fair,
  author  = {Grari, Vincent and Lamprier, Sylvain and Detyniecki, Marcin},
  title   = {A Fair Pricing Model via Adversarial Learning},
  journal = {arXiv preprint},
  year    = {2022},
  note    = {arXiv:2202.12008}
}

@article{araiza2022discrimination,
  author  = {Araiza Iturria, Carlos Andr\'{e}s and Hardy, Mary
             and Marriott, Paul},
  title   = {A Discrimination-Free Premium under a Causal Framework},
  journal = {North American Actuarial Journal},
  volume  = {28},
  number  = {4},
  pages   = {801--821},
  year    = {2024},
  doi     = {10.1080/10920277.2023.2291524}
}

@article{cote2024fair,
  author  = {C\^{o}t\'{e}, Olivier and C\^{o}t\'{e}, Marie-Pier
             and Charpentier, Arthur},
  title   = {A Fair Price to Pay: Exploiting Causal Graphs for
             Fairness in Insurance},
  journal = {Journal of Risk and Insurance},
  year    = {2025},
  doi     = {10.1111/jori.12503}
}

@article{HENCKAERTS2022117230,
  author  = {Henckaerts, Roel and Antonio, Katrien},
  title   = {The Added Value of Dynamically Updating Motor Insurance
             Prices with Telematics Collected Driving Behavior Data},
  journal = {Insurance: Mathematics and Economics},
  volume  = {107},
  pages   = {79--95},
  year    = {2022},
  doi     = {10.1016/j.insmatheco.2022.03.011}
}

@inproceedings{kleinberg-etal-2018-inherent,
  author    = {Kleinberg, Jon and Mullainathan, Sendhil and Raghavan, Manish},
  title     = {Inherent Trade-Offs in the Fair Determination of Risk Scores},
  booktitle = {Proceedings of the 8th Innovations in Theoretical Computer
               Science Conference (ITCS 2017)},
  series    = {Leibniz International Proceedings in Informatics},
  volume    = {67},
  pages     = {43:1--43:23},
  year      = {2017},
  doi       = {10.4230/LIPIcs.ITCS.2017.43}
}

@article{fu-aseri-2022-unfair,
  author  = {Fu, Runshan and Aseri, Manmohan and Singh, Param Vir and Srinivasan, Kannan},
  title   = {{{``Un''}Fair Machine Learning Algorithms}},
  journal = {Management Science},
  volume  = {68},
  number  = {6},
  pages   = {4173--4195},
  year    = {2022},
  doi     = {10.1287/mnsc.2021.4065}
}

@article{shimao-etal-2025-strategic,
  author  = {Shimao, Hajime and Khern-Am-Nuai, Warut and Kannan, Karthik and Cohen, Maxime C.},
  title   = {Strategic Best-Response Fairness Framework for Fair Machine Learning},
  journal = {Information Systems Research},
  volume  = {36},
  number  = {4},
  pages   = {2391--2403},
  year    = {2025},
  doi     = {10.1287/isre.2022.0055}
}
 
\appendix
 
\section{Proofs of Propositions}
\label{app:proofs}
 
\subsection{Proof of Proposition~\ref{prop:ols}}
\label{proof:ols}
 
Write $\hat\beta - \beta^* = (n^{-1} X^\top X)^{-1}
n^{-1} \sum_{i=1}^n x_i r_i$, where $r_i = f(z_i) -
x_i^\top\beta^*$. By Assumption~\ref{ass:sampling}, the
observations are i.i.d., so by the law of large numbers
$n^{-1}X^\top X \plim \Sigma_{xx}$ and
$n^{-1}\sum x_i r_i \plim \E[x r] = 0$ (since $\beta^*$
minimises $\E[(f(z)-x^\top\beta)^2]$, the first-order
condition gives $\E[xr] = 0$). Consistency follows by
Slutsky's theorem and Assumption~\ref{ass:ident}.
 
For the CLT, the summands $\{x_i r_i\}$ are i.i.d.\ with
mean zero and covariance $\Omega = \E[x x^\top r^2]$, which
is finite by Assumptions~\ref{ass:sampling}
and~\ref{ass:approx} (since $\E[\|x\|^2 r^2] \leq
(\E[\|x\|^4])^{1/2}(\E[r^4])^{1/2} < \infty$ by
Cauchy--Schwarz and bounded moments). By the multivariate
CLT, $n^{-1/2}\sum x_i r_i \dlim \mathcal{N}(0,\Omega)$,
and Slutsky's theorem gives the stated sandwich distribution.
The HC0 estimator
$\widehat{\Cov}(\hat\beta) = (X^\top X)^{-1}(\sum_i x_i
x_i^\top \hat r_i^2)(X^\top X)^{-1}$
is consistent because $\hat r_i \to r_i$ in probability
(by consistency of $\hat\beta$) and the continuous mapping
theorem applies under the finite moment conditions.
\hfill$\square$
 
\subsection{Proof of Proposition~\ref{prop:glm}}
\label{proof:glm}
 
Let $s_i(\beta) = \partial \log p(f(z_i)\mid x_i^\top\beta)/
\partial\beta$ denote the score contribution for observation
$i$. The GLM estimator solves $\sum_i s_i(\hat\beta) = 0$.
Under the stated regularity conditions (i.i.d.\ sample, twice
differentiable log-likelihood, full-rank information),
a standard M-estimator argument gives consistency of
$\hat\beta$ for $\beta^*$ and the sandwich CLT
$\sqrt{n}(\hat\beta - \beta^*) \dlim \mathcal{N}(0, J^{-1}MJ^{-1})$,
where $J = -\E[\partial s_i(\beta^*)/\partial\beta^\top]$
is the negative expected Hessian and
$M = \E[s_i(\beta^*)s_i(\beta^*)^\top]$ is the outer product
of scores. Because the response $f(z_i)$ is a deterministic
function rather than a genuine draw from the posited family,
$M \neq J$ in general (the information equality fails), and the standard formula $(X^\top \hat{\Lambda} X)^{-1}$ is therefore incorrect. The sandwich estimator
$\hat J^{-1}\hat M \hat J^{-1}$ is consistent for
$J^{-1}MJ^{-1}$ by the same argument as
Section~\ref{proof:ols}. \hfill$\square$
 
\subsection{Proof of Proposition~\ref{prop:shift}}
\label{proof:shift}
 
From~\eqref{eq:two-estimates} and~\eqref{eq:diff-linear},
$\hat\phi - \hat\phi' = (a - \tilde a)^\top F$, where
$a^\top = e_j^\top(X^\top X)^{-1}X^\top$ and
$\tilde a^\top = e_k^\top(\tilde X^\top\tilde X)^{-1}
\tilde X^\top$ are deterministic functions of the covariate
matrices. The variance across repeated samples from
$\mathbb{P}$ is
\[
  \Var\!\left((a-\tilde a)^\top F\right)
    = a^\top\Cov(F)\,a
      - 2\,a^\top\Cov(F)\,\tilde a
      + \tilde a^\top\Cov(F)\,\tilde a.
\]
Substituting $a^\top = e_j^\top(X^\top X)^{-1}X^\top$ and
$\Cov(F) = \Cov(f(Z))$, the first term equals the $j$-th
diagonal element of $(X^\top X)^{-1}\Cov(X f(Z))
(X^\top X)^{-1}$, which is the sandwich variance of
$\hat\phi$ from the restricted model. The third term is
the analogous sandwich variance of $\hat\phi'$. The cross
term is
\[
  a^\top\Cov(F)\,\tilde a
    = e_j^\top(X^\top X)^{-1}
      \Cov(X_j f(Z),\,\tilde X_k f(Z))
      (\tilde X^\top\tilde X)^{-1}e_k,
\]
which is non-zero whenever $\Cov(X_j f(Z), \tilde X_k
f(Z)) \neq 0$, that is, whenever the two covariate-response
products are correlated across the sample. Since both
models use the same $f(z_i)$, this correlation is
generically non-zero. The sample estimators stated in
Proposition~\ref{prop:shift} are the plug-in versions of
these population quantities, and their consistency follows
from Assumption~\ref{ass:sampling} and the law of large
numbers. \hfill$\square$
 
\section{Aligning Proxy Discrimination and Conditional Demographic Parity with Regulatory Frameworks} \label{app:PD_CDP}

\subsection{Proxy discrimination} \label{app:pd}

\emph{Underlying criterion}: observed rating variables should not serve as statistical substitutes for a protected attribute. When a variable correlates with $A$ and absorbs some of $A$'s predictive power for $P$, removing $A$ from the model does not eliminate its influence on prices. Instead, that influence passes through the correlated variable.

Proxy discrimination has a straightforward econometric structure. It can be viewed as a special case of omitted-variable bias in which the omitted variable is the protected characteristic $A$. Let $P$ denote the pricing outcome, $X_\ell$ a set of legitimate rating factors, and $W$ a set of additional (potentially non-traditional) variables used in the model. When $A$ is excluded from a regression of $P$ on $X_\ell$ and $W$, the estimated coefficients $\hat{\phi}$ on $W$ will be biased to the extent that $W$ is correlated with $A$ and $A$ has a causal effect on $P$. Re-introducing $A$ into the regression and observing a shift in $\hat{\phi}$ therefore provides direct evidence that $W$ is acting as a proxy for the protected attribute \citep{pope-sydnor-2011-implementing, lindholm-etal-2022-discrimination}.

Draft regulation by the \citet[][\S8]{doradoi-2023-concerning} (CO DOI) operationalises this insight, though without naming its connection to omitted-variable bias. The procedure compares two regression models:
\begin{enumerate}
    \item Fit two models for the outcome of interest. The first regresses the outcome on (a)~traditional underwriting factors, (b)~non-traditional variables used in the pricing or approval decision, and (c)~race/ethnicity indicator variables. The second regresses the outcome on sets (a) and (b) only, omitting (c). CO DOI recommend logistic regression for binary approval outcomes and linear regression for premium rates per \$1{,}000 of face amount. Race/ethnicity is estimated via BIFSG (Bayesian Improved First Name
Surname Geocoding, a method that infers the probability of an
individual's race or ethnicity from their first name, surname, and
residential geography) \citep{voicu2018using} when not directly observed.

    \item Examine the coefficients on the non-traditional variables (set b) across the two models. Under CO DOI's draft regulation, any variable whose coefficient shifts between the two models is taken as evidence that the variable may contribute to unfair discrimination.
\end{enumerate}

The same logic was applied by the \citet{ftc-2007-creditBased} to investigate whether credit-based insurance scores proxy for race/ethnicity and neighbourhood income in automobile claims-cost models, giving the CO DOI approach a precedent in federal regulatory practice.

A complementary approach, proposed by \citet{duPreez-etal-2024-bias}, builds a classifier that predicts $A$ from the non-protected variables. Strong predictive accuracy implies that those variables can substitute for $A$, and thus that proxy discrimination is plausible. This approach is less model-dependent than the coefficient-shift method but yields a weaker conclusion: it establishes potential for proxy discrimination rather than measuring its extent in a specific pricing model.

The \citet{nydfs-2024-proposed} (NY DFS) proposes a related concept under the label ‘drivers of disparity’. NY DFS requires insurers to identify variables that “cause differences in outcomes for protected classes relative to control groups” (§17.vi) and to demonstrate that observed characteristics do not “serve as a proxy for any protected classes that may result in unfair or unlawful discrimination” (§11). We interpret this as aligning with the proxy discrimination framework described above. In particular, the requirement to assess whether variables act as proxies can be operationalised through a coefficient-shift analysis, in which changes in estimated effects after controlling for the protected attribute provide evidence of proxy behaviour. Under this interpretation, the CO DOI and NY DFS approaches are substantively aligned despite their different terminology.

\subsection{Conditional Demographic Parity} 
\label{app:cdp}
\emph{Underlying criterion}: after accounting for legitimate risk differences, members of different groups should face the same distribution of prices. CDP is the conditional analogue of demographic parity (DP or independence criterion). It permits price differences explained by $X_\ell$ but not residual differences attributable to group membership alone \citep{xin-huang-2023-anti}.

Formally, CDP holds for a pricing outcome $P$ if
$$\Pr\!\left(P = p \mid X_\ell = x_\ell, A = a\right) =
  \Pr\!\left(P = p \mid X_\ell = x_\ell, A = b\right)$$
for all $p$ and all values $X_\ell$ of the legitimate rating
factors $X_\ell$.

In its strict form, CDP requires equality of the entire conditional
distribution of prices across groups. In practice, and throughout
this paper, we operationalise a mean-based relaxation: we test
whether the conditional expectation $\mathbb{E}[P = p \mid X_{\ell} = x_\ell, A = a]$
equals $\mathbb{E}[P = p \mid X_{\ell} = x_\ell, A = b]$, which is the quantity
identified by the regression coefficient $\beta$. This is the
formulation adopted in the major regulatory proposals.

A ratio relaxation, relaxed conditional demographic parity (RCDP) at tolerance level $\tau$, permits a bounded disparity:
$$\tau \;\leq\; \frac{\Pr\!\left(P = p \mid X_{\ell} = x_\ell, A = a\right)}
                     {\Pr\!\left(P = p \mid X_{\ell} = x_\ell, A = b\right)}
              \;\leq\; \frac{1}{\tau}$$
The $\tau$ threshold corresponds directly to the adverse impact ratio used in employment discrimination law, so RCDP is a natural operationalisation for regulators already working within that framework. When the conditioning on legitimate rating factors $X_\ell$ is removed, the criterion reduces to demographic parity (DP), which requires equality in the marginal distribution of outcomes across groups, without adjusting for differences in underlying risk.

In its strict form, RCDP requires this ratio to lie within
$(\tau, 1/\tau)$ for all $p$ and all values $x_\ell$ of
the legitimate rating factors, which is a pointwise condition that
is generally untestable from finite data. In practice, the regression imposes the assumption that the
conditional mean gap is constant across all values of
$X_\ell$. Under this assumption, the coefficient
estimates a single gap that applies uniformly after
controlling for $X_\ell$, rather than a pointwise condition
at each $x_\ell$. 

The CO DOI draft regulation \citep[][\S\S\,6--7]{doradoi-2023-concerning} encodes a two-step test that we interpret as operationalising CDP and RCDP, respectively, though the regulation does not use this terminology.\footnote{The mapping from the CO DOI procedure to CDP and RCDP is our own interpretation, based on the algebraic structure of the proposed tests.}
\begin{enumerate}
    \item Regress the outcome variable on a set of race/ethnicity indicator variables and, optionally, a limited set of approved control variables. CO DOI recommend logistic regression for binary approval outcomes and linear regression for premium rates per \$1{,}000 of face amount; unobserved race/ethnicity is estimated using BIFSG.

    \item \textbf{First test (CDP):} assess whether the race/ethnicity indicators are jointly or individually significant at the 5\% level. A model passes this test if none of the race/ethnicity coefficients is statistically significant.

    \item \textbf{Second test (RCDP):} for any significant race/ethnicity indicator, assess whether its estimated effect exceeds the regulatory tolerance. CO DOI set this at 5 percentage points for approval rates and 5\% of the mean premium for price outcomes. A model passes this test if all significant effects are below these thresholds.
\end{enumerate}

When no control variables are included in the regression, the same two tests instead probe the unconditional criteria, demographic independence and the raw demographic impact ratio, respectively. Including controls shifts the comparison to the conditional (CDP/RCDP) setting.

A fundamental limitation of the CO DOI procedure is that its pass/fail rule conflates statistical significance with practical importance in a way that is sensitive to sample size. In large datasets, even negligible disparities will be flagged as statistically significant, while in small datasets, substantively large disparities may go undetected. The TOST framework addresses this limitation by treating the tolerance margin as the primary inferential target.

Several additional metrics proposed by the NY DFS \citep{nydfs-2024-proposed} are consistent with testing for independence (demographic parity) or CDP; these include the Adverse Impact Ratio, Denials Odds Ratio, Standardised Mean Differences, and $z$/$t$-tests.

Each metric can be computed against observed outcomes (corresponding to DP) or against regression residuals that control for $X_\ell$ (corresponding to CDP). The NY DFS guidance does not resolve which mode is required, leaving the choice to the insurer. Our framework addresses this by tying the conditioning set to the pre-specified list of legitimate rating factors $X_\ell$ .

\section{Supplementary Empirical Results}
\label{app:empirical}
 
\subsection{CDP regression: classical versus HC3 standard
errors, all 34 companies}
\label{app:empirical:se}
 
\begin{center}
\begin{longtable}{lrrrrr}
\caption{CDP regression: classical versus HC3 standard errors
  for the minority coefficient $\hat\beta_{A,k}$, all 34 Illinois
  companies. $\rho = \mathrm{SE}_{\mathrm{HC3}} /
  \mathrm{SE}_{\mathrm{classical}}$; $R^2$ is the fit of the
  linear audit model. Sorted by $\rho$ (largest first).}
\label{tab:se-full} \\
\hline
\textbf{Company} & $\hat\beta_A$ & $\mathrm{SE}_{\mathrm{cl}}$ &
  $\mathrm{SE}_{\mathrm{HC3}}$ & $\rho$ & $R^2$ \\
\hline
\endfirsthead
\multicolumn{6}{c}{\tablename\ \thetable\, continued.} \\
\hline
\textbf{Company} & $\hat\beta_A$ & $\mathrm{SE}_{\mathrm{cl}}$ &
  $\mathrm{SE}_{\mathrm{HC3}}$ & $\rho$ & $R^2$ \\
\hline
\endhead
\hline
\endlastfoot
Economy Preferred Ins Co          & +0.308 & 0.017 & 0.030 & 1.775 & 0.476 \\
Metropolitan Cas Ins Co           & +0.308 & 0.017 & 0.029 & 1.679 & 0.482 \\
Farmers Automobile Ins Assoc      & +0.294 & 0.024 & 0.040 & 1.660 & 0.675 \\
Travelers Home \& Marine Ins Co   & +0.209 & 0.015 & 0.020 & 1.330 & 0.546 \\
Travelers Commercial Ins Co       & +0.209 & 0.015 & 0.020 & 1.327 & 0.546 \\
Owners Ins Co                     & +0.314 & 0.020 & 0.025 & 1.246 & 0.554 \\
Country Mut Ins Co                & +0.260 & 0.017 & 0.021 & 1.192 & 0.349 \\
Country Pref Ins Co               & +0.258 & 0.017 & 0.020 & 1.190 & 0.353 \\
Metropolitan Prop \& Cas Ins Co   & +0.356 & 0.028 & 0.032 & 1.144 & 0.475 \\
Allstate Fire \& Cas Ins Co       & +0.215 & 0.016 & 0.019 & 1.123 & 0.437 \\
Metropolitan Grp Prop \& Cas Ins Co & +0.280 & 0.022 & 0.025 & 1.121 & 0.414 \\
Allstate Ind Co                   & +0.318 & 0.022 & 0.024 & 1.094 & 0.507 \\
Government Employees Ins Co       & +0.165 & 0.013 & 0.014 & 1.084 & 0.342 \\
Geico Gen Ins Co                  & +0.165 & 0.013 & 0.014 & 1.084 & 0.342 \\
Erie Ins Exch                     & +0.221 & 0.015 & 0.017 & 1.075 & 0.480 \\
Erie Ins Co                       & +0.221 & 0.015 & 0.016 & 1.072 & 0.480 \\
Geico Ind Co                      & +0.172 & 0.013 & 0.014 & 1.057 & 0.369 \\
First Liberty Ins Corp            & +0.182 & 0.013 & 0.014 & 1.046 & 0.504 \\
Liberty Mut Fire Ins Co           & +0.182 & 0.013 & 0.014 & 1.042 & 0.502 \\
Safeco Ins Co Of IL               & +0.141 & 0.010 & 0.010 & 1.019 & 0.457 \\
United Serv Automobile Assn       & +0.095 & 0.010 & 0.010 & 0.995 & 0.324 \\
USAA Cas Ins Co                   & +0.090 & 0.010 & 0.009 & 0.955 & 0.300 \\
Illinois Farmers Ins Co           & +0.171 & 0.018 & 0.017 & 0.940 & 0.359 \\
USAA Gen Ind Co                   & +0.114 & 0.010 & 0.009 & 0.907 & 0.401 \\
Garrison Prop \& Cas Ins Co       & +0.087 & 0.012 & 0.010 & 0.873 & 0.244 \\
American Family Mut Ins Co        & +0.151 & 0.012 & 0.011 & 0.870 & 0.502 \\
Progressive Direct Ins Co         & +0.157 & 0.013 & 0.011 & 0.870 & 0.417 \\
Progressive Universal Ins Co      & +0.150 & 0.012 & 0.011 & 0.869 & 0.419 \\
Progressive Northern Ins Co       & +0.213 & 0.018 & 0.015 & 0.835 & 0.417 \\
American Standard Ins Co of WI    & +0.137 & 0.013 & 0.010 & 0.801 & 0.431 \\
Geico Cas Co                      & +0.159 & 0.012 & 0.009 & 0.753 & 0.424 \\
State Farm Mut Auto Ins Co        & +0.222 & 0.017 & 0.013 & 0.753 & 0.450 \\
State Farm Fire \& Cas Co         & +0.222 & 0.017 & 0.013 & 0.753 & 0.450 \\
Trumbull Ins Co                   & +0.130 & 0.013 & 0.009 & 0.685 & 0.405 \\
\end{longtable}
\end{center}
 
\subsection{Full CDP audit results, all 34 companies}
\label{app:empirical:cdp}
 
\begin{center}
\begin{longtable}{lrrrr}
\caption{Full CDP audit results, all 34 Illinois companies.
  $\beta_A$ = conditional log-premium gap (HC3 SEs);
  Ratio $= e^{\beta_A}$; Gap = implied dollar difference at mean
  premium (\$370). TOST: $\delta = 5\%$ of mean, $\tau = 0.80$,
  $\alpha = 0.10$. All companies fail.}
\label{tab:cdp-full} \\
\hline
\textbf{Company} & \textbf{Gap (\$)} & \textbf{Ratio} &
  \textbf{90\% CI for $\beta_A$} & \textbf{Dec.} \\
\hline
\endfirsthead
\multicolumn{5}{c}{\tablename\ \thetable\ continued.} \\
\hline
\textbf{Company} & \textbf{Gap (\$)} & \textbf{Ratio} &
  \textbf{90\% CI for $\beta_A$} & \textbf{Dec.} \\
\hline
\endhead
\hline
\endlastfoot
Metropolitan Prop \& Cas Ins Co   & \$158 & 1.427 &
  $[+0.315,\;+0.397]$ & FAIL \\
Allstate Ind Co                   & \$138 & 1.374 &
  $[+0.287,\;+0.349]$ & FAIL \\
Owners Ins Co                     & \$137 & 1.369 &
  $[+0.282,\;+0.346]$ & FAIL \\
Economy Preferred Ins Co          & \$134 & 1.361 &
  $[+0.269,\;+0.346]$ & FAIL \\
Metropolitan Cas Ins Co           & \$134 & 1.361 &
  $[+0.271,\;+0.345]$ & FAIL \\
Farmers Automobile Ins Assoc      & \$127 & 1.342 &
  $[+0.243,\;+0.345]$ & FAIL \\
Metropolitan Grp Prop \& Cas Ins Co & \$120 & 1.323 &
  $[+0.249,\;+0.311]$ & FAIL \\
Country Mut Ins Co                & \$110 & 1.297 &
  $[+0.234,\;+0.287]$ & FAIL \\
Country Pref Ins Co               & \$109 & 1.295 &
  $[+0.233,\;+0.284]$ & FAIL \\
Erie Ins Exch                     & \$ 92 & 1.248 &
  $[+0.200,\;+0.242]$ & FAIL \\
State Farm Fire \& Cas Co         & \$ 92 & 1.248 &
  $[+0.205,\;+0.238]$ & FAIL \\
State Farm Mut Auto Ins Co        & \$ 92 & 1.248 &
  $[+0.205,\;+0.238]$ & FAIL \\
Erie Ins Co                       & \$ 91 & 1.247 &
  $[+0.200,\;+0.242]$ & FAIL \\
Allstate Fire \& Cas Ins Co       & \$ 89 & 1.240 &
  $[+0.191,\;+0.239]$ & FAIL \\
Progressive Northern Ins Co       & \$ 88 & 1.238 &
  $[+0.194,\;+0.232]$ & FAIL \\
Travelers Home \& Marine Ins Co   & \$ 86 & 1.233 &
  $[+0.184,\;+0.235]$ & FAIL \\
Travelers Commercial Ins Co       & \$ 86 & 1.232 &
  $[+0.183,\;+0.234]$ & FAIL \\
Liberty Mut Fire Ins Co           & \$ 74 & 1.199 &
  $[+0.164,\;+0.200]$ & FAIL \\
First Liberty Ins Corp            & \$ 74 & 1.199 &
  $[+0.164,\;+0.199]$ & FAIL \\
Geico Ind Co                      & \$ 70 & 1.188 &
  $[+0.154,\;+0.190]$ & FAIL \\
Illinois Farmers Ins Co           & \$ 69 & 1.186 &
  $[+0.149,\;+0.193]$ & FAIL \\
Geico Gen Ins Co                  & \$ 66 & 1.179 &
  $[+0.147,\;+0.183]$ & FAIL \\
Government Employees Ins Co       & \$ 66 & 1.179 &
  $[+0.147,\;+0.183]$ & FAIL \\
Geico Cas Co                      & \$ 64 & 1.173 &
  $[+0.148,\;+0.171]$ & FAIL \\
Progressive Direct Ins Co         & \$ 63 & 1.170 &
  $[+0.143,\;+0.171]$ & FAIL \\
American Family Mut Ins Co        & \$ 60 & 1.163 &
  $[+0.137,\;+0.164]$ & FAIL \\
Progressive Universal Ins Co      & \$ 60 & 1.162 &
  $[+0.136,\;+0.163]$ & FAIL \\
Safeco Ins Co Of IL               & \$ 56 & 1.152 &
  $[+0.128,\;+0.154]$ & FAIL \\
American Standard Ins Co of WI    & \$ 54 & 1.147 &
  $[+0.124,\;+0.150]$ & FAIL \\
Trumbull Ins Co                   & \$ 51 & 1.138 &
  $[+0.119,\;+0.141]$ & FAIL \\
USAA Gen Ind Co                   & \$ 44 & 1.120 &
  $[+0.102,\;+0.126]$ & FAIL \\
United Serv Automobile Assn       & \$ 37 & 1.099 &
  $[+0.082,\;+0.107]$ & FAIL \\
USAA Cas Ins Co                   & \$ 35 & 1.095 &
  $[+0.079,\;+0.102]$ & FAIL \\
Garrison Prop \& Cas Ins Co       & \$ 34 & 1.091 &
  $[+0.074,\;+0.100]$ & FAIL \\
\end{longtable}
\end{center}
 
\subsection{Full proxy discrimination results, all 34 companies}
\label{app:empirical:proxy}
 
Both standard errors use the score-based sandwich estimator
appropriate for deterministic responses
(Propositions~\ref{prop:ols} and~\ref{prop:shift}).
$\mathrm{SE}_{\mathrm{ind}}$ sums the individual
score-based sandwich variances treating the two models as
independent; $\mathrm{SE}_{\mathrm{full}}$ applies the
cross-covariance correction of Proposition~\ref{prop:shift}.
$^{*}$ denotes significance at 5\% (two-sided,
$|z| > 1.645$). Sorted by ratio (smallest first).
 
\begin{center}
\begin{longtable}{lrrrrrr}
\caption{Proxy discrimination test: coefficient shift on log
  state risk when minority flag is added, all 34 Illinois
  companies. All standard errors use the score-based sandwich.
  Ratio $= \mathrm{SE}_{\mathrm{full}}/\mathrm{SE}_{\mathrm{ind}}$.
  Mean ratio $= 0.082$; range $[0.080, 0.085]$.}
\label{tab:proxy-full} \\
\hline
\textbf{Company} & \textbf{Shift} &
  $\mathrm{SE}_{\mathrm{ind}}$ &
  $\mathrm{SE}_{\mathrm{full}}$ & \textbf{Ratio} &
  $z_{\mathrm{ind}}$ & $z_{\mathrm{full}}$ \\
\hline
\endfirsthead
\multicolumn{7}{c}{\tablename\ \thetable\ continued.} \\
\hline
\textbf{Company} & \textbf{Shift} &
  $\mathrm{SE}_{\mathrm{ind}}$ &
  $\mathrm{SE}_{\mathrm{full}}$ & \textbf{Ratio} &
  $z_{\mathrm{ind}}$ & $z_{\mathrm{full}}$ \\
\hline
\endhead
\hline
\endlastfoot
  Garrison Prop \& Cas Ins Co                   & $7.8$\% & $0.0250$ & $0.0020$ & $0.080$ & $0.32$ & $4.00$\(^{*}\) \\
  USAA Cas Ins Co                               & $11.5$\% & $0.0229$ & $0.0018$ & $0.080$ & $0.37$ & $4.53$\(^{*}\) \\
  Geico Ind Co                                  & $10.2$\% & $0.0276$ & $0.0022$ & $0.081$ & $0.58$ & $7.09$\(^{*}\) \\
  American Standard Ins Co of WI                & $8.0$\% & $0.0285$ & $0.0023$ & $0.081$ & $0.44$ & $5.46$\(^{*}\) \\
  Geico Cas Co                                  & $8.5$\% & $0.0249$ & $0.0020$ & $0.081$ & $0.59$ & $7.25$\(^{*}\) \\
  Government Employees Ins Co                   & $11.1$\% & $0.0245$ & $0.0020$ & $0.081$ & $0.62$ & $7.65$\(^{*}\) \\
  Progressive Direct Ins Co                     & $7.2$\% & $0.0268$ & $0.0022$ & $0.081$ & $0.54$ & $6.65$\(^{*}\) \\
  Geico Gen Ins Co                              & $11.1$\% & $0.0245$ & $0.0020$ & $0.081$ & $0.62$ & $7.65$\(^{*}\) \\
  Safeco Ins Co Of IL                           & $10.5$\% & $0.0251$ & $0.0020$ & $0.081$ & $0.52$ & $6.42$\(^{*}\) \\
  Progressive Universal Ins Co                  & $7.2$\% & $0.0261$ & $0.0021$ & $0.081$ & $0.53$ & $6.51$\(^{*}\) \\
  United Serv Automobile Assn                   & $11.1$\% & $0.0228$ & $0.0018$ & $0.081$ & $0.38$ & $4.75$\(^{*}\) \\
  USAA Gen Ind Co                               & $9.0$\% & $0.0247$ & $0.0020$ & $0.081$ & $0.43$ & $5.25$\(^{*}\) \\
  First Liberty Ins Corp                        & $9.5$\% & $0.0257$ & $0.0021$ & $0.082$ & $0.65$ & $7.95$\(^{*}\) \\
  Liberty Mut Fire Ins Co                       & $9.5$\% & $0.0256$ & $0.0021$ & $0.082$ & $0.66$ & $7.99$\(^{*}\) \\
  Illinois Farmers Ins Co                       & $8.7$\% & $0.0259$ & $0.0021$ & $0.082$ & $0.61$ & $7.44$\(^{*}\) \\
  Trumbull Ins Co                               & $6.6$\% & $0.0229$ & $0.0019$ & $0.082$ & $0.52$ & $6.39$\(^{*}\) \\
  American Family Mut Ins Co                    & $7.7$\% & $0.0261$ & $0.0021$ & $0.082$ & $0.53$ & $6.53$\(^{*}\) \\
  Allstate Fire \& Cas Ins Co                   & $9.7$\% & $0.0262$ & $0.0022$ & $0.082$ & $0.76$ & $9.18$\(^{*}\) \\
  State Farm Mut Auto Ins Co                    & $7.4$\% & $0.0249$ & $0.0021$ & $0.083$ & $0.82$ & $9.91$\(^{*}\) \\
  State Farm Fire \& Cas Co                     & $7.4$\% & $0.0254$ & $0.0021$ & $0.083$ & $0.81$ & $9.75$\(^{*}\) \\
  Erie Ins Co                                   & $9.3$\% & $0.0251$ & $0.0021$ & $0.083$ & $0.81$ & $9.80$\(^{*}\) \\
  Country Pref Ins Co                           & $10.6$\% & $0.0243$ & $0.0020$ & $0.083$ & $0.98$ & $11.83$\(^{*}\) \\
  Travelers Home \& Marine Ins Co               & $14.6$\% & $0.0248$ & $0.0020$ & $0.083$ & $0.78$ & $9.46$\(^{*}\) \\
  Erie Ins Exch                                 & $9.3$\% & $0.0245$ & $0.0020$ & $0.083$ & $0.84$ & $10.08$\(^{*}\) \\
  Travelers Commercial Ins Co                   & $14.6$\% & $0.0246$ & $0.0020$ & $0.083$ & $0.79$ & $9.51$\(^{*}\) \\
  Progressive Northern Ins Co                   & $7.0$\% & $0.0256$ & $0.0021$ & $0.083$ & $0.77$ & $9.33$\(^{*}\) \\
  Country Mut Ins Co                            & $10.6$\% & $0.0244$ & $0.0020$ & $0.083$ & $0.98$ & $11.87$\(^{*}\) \\
  Metropolitan Grp Prop \& Cas Ins Co           & $9.3$\% & $0.0268$ & $0.0022$ & $0.083$ & $0.96$ & $11.59$\(^{*}\) \\
  Owners Ins Co                                 & $11.4$\% & $0.0255$ & $0.0022$ & $0.084$ & $1.14$ & $13.48$\(^{*}\) \\
  Metropolitan Cas Ins Co                       & $15.5$\% & $0.0258$ & $0.0022$ & $0.084$ & $1.10$ & $13.20$\(^{*}\) \\
  Allstate Ind Co                               & $10.5$\% & $0.0277$ & $0.0023$ & $0.084$ & $1.06$ & $12.60$\(^{*}\) \\
  Economy Preferred Ins Co                      & $16.5$\% & $0.0248$ & $0.0021$ & $0.084$ & $1.15$ & $13.71$\(^{*}\) \\
  Metropolitan Prop \& Cas Ins Co               & $12.5$\% & $0.0245$ & $0.0021$ & $0.085$ & $1.34$ & $15.73$\(^{*}\) \\
  Farmers Automobile Ins Assoc                  & $22.2$\% & $0.0253$ & $0.0022$ & $0.085$ & $1.07$ & $12.59$\(^{*}\) \\
\hline
\multicolumn{4}{l}{Mean ratio $= 0.082$ \quad Range $= [0.080, 0.085]$}
  & & & \\
\hline
\end{longtable}
\end{center}
 
\end{document}